\documentclass[12pt]{article}
\usepackage{graphicx,amsfonts,amsmath,amssymb,amsthm,mathrsfs,url,hyperref}
\usepackage[usenames,dvipsnames]{xcolor}

\title{Interior-Boundary Conditions for Many-Body Dirac Operators and Codimension-1 Boundaries}
\author{
Julian Schmidt,\footnote{Mathematisches Institut,
     Eberhard-Karls-Universit\"at, Auf der Morgenstelle 10, 72076
     T\"ubingen, Germany}\ \footnote{E-mail:
     juls@maphy.uni-tuebingen.de}\ \ 
Stefan Teufel,$^*$\footnote{E-mail: stefan.teufel@uni-tuebingen.de}\ \  and 
Roderich Tumulka$^*$\footnote{E-mail: roderich.tumulka@uni-tuebingen.de}
}
\date{March 26, 2019}

\newcommand{\Hilbert}{\mathscr{H}}

\newcommand{\Q}{\mathcal{Q}}
\newcommand{\sM}{\mathscr{M}}
\renewcommand{\Re}{\mathrm{Re}}
\renewcommand{\Im}{\mathrm{Im}}
\newcommand{\RRR}{\mathbb{R}}
\newcommand{\CCC}{\mathbb{C}}

\newcommand{\scp}[2]{\langle #1|#2 \rangle}

\newcommand{\Lz}{L^2} 
\newcommand{\B}{\mathfrak{B}}
\newcommand{\tr}{\mathrm{tr}}

\newcommand{\vj}{\boldsymbol{j}}
\newcommand{\vk}{\boldsymbol{k}}
\newcommand{\vn}{\boldsymbol{n}}

\newcommand{\vt}{\boldsymbol{t}}

\newcommand{\vx}{\boldsymbol{x}}

\newcommand{\vN}{\boldsymbol{N}}

\newcommand{\valpha}{\boldsymbol{\alpha}}

\newcommand{\vsigma}{\boldsymbol{\sigma}}

\newcommand{\vzero}{\boldsymbol{0}}

\newcommand{\tpsi}{\widetilde{\psi}}
\newcommand{\tA}{\widetilde{A}}
\newcommand{\tE}{\widetilde{E}}
\newcommand{\tF}{\widetilde{F}}
\newcommand{\tL}{\widetilde{L}}

\newcommand{\tS}{\widetilde{S}}

\newcommand{\hA}{\widehat{A}}
\newcommand{\hE}{\widehat{E}}

\newcommand{\hP}{\widehat{P}}

\DeclareMathOperator{\diag}{diag}
\newcommand{\domain}{\mathscr{D}}
\newcommand{\Kegel}{\mathscr{K}}
\newcommand{\tKegel}{\widetilde{\mathscr{K}}}

\newtheorem{thm}{Theorem}

\newtheorem{prop}{Proposition}
\newtheorem{conj}{Conjecture}
\newcommand{\be}{\begin{equation}}
\newcommand{\ee}{\end{equation}}
\newcounter{remarks}
\begin{document}
\maketitle
\begin{abstract}
We are dealing with boundary conditions for Dirac-type operators, i.e., first order differential operators with matrix-valued coefficients, including in particular physical many-body Dirac operators. We characterize (what we conjecture is) the general form of reflecting boundary conditions (which includes known boundary conditions such as the one of the MIT bag model) and, as our main goal, of interior-boundary conditions (IBCs). IBCs are a new approach to defining UV-regular Hamiltonians for quantum field theories without smearing particles out or discretizing space. For obtaining such Hamiltonians, the method of IBCs provides an alternative to renormalization and has been successfully used so far in non-relativistic models, where it could be applied also in cases in which no renormalization method was known. A natural next question about IBCs is how to set them up for the Dirac equation, and here we take first steps towards the answer. For quantum field theories, the relevant boundary consists of the surfaces in $n$-particle configuration space $\RRR^{3n}$ on which two particles have the same location in $\RRR^3$. While this boundary has codimension 3, we focus here on the more basic situation in which the boundary has codimension 1 in configuration space. We describe specific examples of IBCs for the Dirac equation, we prove for some of these examples that they rigorously define self-adjoint Hamiltonians, and we develop the general form of IBCs for Dirac-type operators.

\medskip

  \noindent 
%
%
\end{abstract}

\section{Introduction}

An interior--boundary condition (IBC) is a condition on a wave function $\psi$ defined on a configuration space $\Q$ with boundaries that relates the values (or derivatives) of $\psi$ on the boundary of $\Q$ to the values of $\psi$ at suitable interior points of $\Q$. Like other boundary conditions, an IBC can accompany a partial differential equation to define the time evolution of $\psi$. For defining Hamiltonians with particle creation and annihilation (such as in quantum field theory), $\psi$ would be a particle--position representation of the quantum state, $\Q$ a configuration space of a variable number of particles, and the IBC would involve two configurations related by a creation or annihilation event, i.e., an $n$-particle configuration with two particles at the same location (a configuration regarded as a boundary point of $\Q$) and the $(n-1)$-particle configuration with one of the two particles removed (which is an interior point of $\Q$).

Remarkably, IBCs allow the definition of the Hamiltonian without an ultraviolet (UV) cut-off or renormalization, and it has been verified in the non-relativistic case that the Hamiltonians thus defined exist rigorously and are self-adjoint \cite{TT15a,ibc2a,Lam18}, even in cases in which no renormalization method was known before \cite{Lam18}. In cases in which renormalized Hamiltonians were known to exist, it turned out that they agree with the IBC-Hamiltonian up to addition of an (irrelevant) finite constant \cite{ibc2a,LS18,Sch18}, which suggests that the IBC-Hamiltonians are physically reasonable. It is thus of interest to develop IBCs also for relativistic equations such as the Dirac equation. 

While the boundaries relevant to particle creation usually have codimension 3, the first and most basic question about IBCs for the Dirac operator would be what these conditions can look like for the simplest kind of boundary, i.e., for a boundary of codimension 1 (as in a half space). On these boundaries we focus here. For example, the ``MIT bag model'' of quark confinement \cite{CJJTW74,CJJT74} involves the Dirac equation with a reflecting boundary condition on the surface of the confinement region, a surface of codimension 1; an IBC on this surface would allow for particle creation and annihilation on this surface. Codimension-1 boundaries provide a natural framework for a theory of IBCs, although not the only possible one.

The purpose of this paper is to develop IBCs suitable for the Dirac equation, whose basic form is
\be\label{Dirac}
i\hbar\frac{\partial \psi}{\partial t} = -i\hbar \valpha\cdot \nabla \psi + m\beta \psi\,,
\ee
where $\psi$ is a $\CCC^4$-valued function on $\RRR^3$, we have set the speed of light $c$ to 1, $m\geq 0$ is the mass, and $\valpha=(\alpha^1,\alpha^2,\alpha^3)$ and $\beta$ are the Dirac alpha and beta matrices. We also consider general first-order differential operators with matrix-valued coefficients, which we call ``Dirac-type'' operators.\footnote{In contrast to some authors (e.g., \cite{BB13}), we do not demand that the coefficients satisfy the Clifford relations because we want to include many-particle Dirac operators. See also Footnote~\ref{fn:Clifford} in Section~\ref{sec:config1}.} We formulate what we conjecture is the general form of IBCs on codimension-1 boundaries for Hamiltonians that are first-order differential operators with matrix-valued coefficients, we derive the probability balance equation for $|\psi|^2$, and we obtain that overall probability is conserved. We thus achieve the analog for first-order differential operators of what was achieved in \cite{co1} for the Laplacian operator. Moreover, for a simple example we provide a rigorous proof that the IBC method leads to a self-adjoint Hamiltonian. Our analysis also yields the general form of \emph{reflecting} boundary conditions (as opposed to \emph{interior}--boundary conditions) for Dirac-type equations; as far as we know, this general form has not been described before, while various special cases are well known, including the boundary condition of the MIT bag model.

The approach of IBCs was discussed extensively in \cite{TT15a,TT15b,ibc2a} after previous pioneering work in \cite{Mosh51a,Mosh51b,Tho84,Yaf92, Tum04}. Bohmian trajectories associated with IBCs (including those presented here) are defined in \cite{bohmibc}. In \cite{timelike}, the IBCs developed here are applied to a model of particle emission and absorption by a naked space-time singularity. Lienert and Nickel \cite{LN18} employ an IBC of the type discussed here for particle creation in one-dimensional space and extend it to a multi-time evolution; in this way, they provide an (almost fully relativistic) rigorously defined quantum model of particle creation. Further work on IBCs can be found in \cite{KS15,Sch18,ST18}. For discussion of various reflecting boundary conditions for the Dirac equation, see, e.g., \cite{BB13,FR15}; for absorbing boundary conditions for the Dirac equation, see \cite{detect-dirac}.

This paper is organized as follows. In Section~\ref{sec:ex1}, we describe an explicit example of an IBC for the Dirac equation on a codimension-1 boundary and state our result about self-adjointness. In Section~\ref{sec:reflect1}, as a preparation for the general IBCs, we discuss the reflecting boundary conditions for Dirac-type equations (many of which are known). In Section~\ref{sec:general}, we develop and discuss the general form of IBCs for Dirac-type equations, including a calculation verifying that probability is conserved. 
Appendices~\ref{app:sa} and \ref{app:proofs} contain the mathematical proofs.

\section{Example of an Interior--Boundary Condition}
\label{sec:ex1}

Before turning to the general theory of all possible IBCs, we formulate a concrete example of an IBC for the Dirac equation. Another example in one space dimension is described in \cite{LN18}.

\subsection{Defining Equations}
\label{sec:defex1}

As a configuration space with boundary, consider $\Q=\Q^{(0)}\cup \Q^{(1)}$ with $\Q^{(0)}=\{\emptyset\}$ containing a single point $\emptyset$ (the ``zero-particle configuration'') and $\Q^{(1)}$ a half space
\be
\Q^{(1)}=\RRR^{3}_{>}=\Bigl\{ (x_1,x_2,x_3)\in\RRR^3: x_3\geq 0 \Bigr\}
\ee
with boundary $\partial \Q^{(1)}$ given by the plane $\{x_3=0\}$. The Hilbert space $\Hilbert$ is a ``mini-Fock space'' $\Hilbert = \Hilbert^{(0)} \oplus \Hilbert^{(1)}$ (orthogonal sum) with $\Hilbert^{(0)}=\CCC$ and $\Hilbert^{(1)}=\Lz(\RRR^3_{>},\CCC^4)$. Elements of $\Hilbert$ can be regarded as functions $\psi$ on $\Q$; that is because any function $\psi$ on a disjoint union $\Q^{(0)}\cup \Q^{(1)}$ consists of a function $\psi^{(0)}$ on $\Q^{(0)}$ and a function $\psi^{(1)}$ on $\Q^{(1)}$, here $\psi^{(0)}: \Q^{(0)}\to \CCC$ and $\psi^{(1)}:\Q^{(1)}\to \CCC^4$; since $\Q^{(0)}$ has merely a single element, $\psi^{(0)}$ can be identified with the complex number that is the value at that single element, so that $\psi^{(0)}\in \Hilbert^{(0)}$ and $\psi^{(1)}\in\Hilbert^{(1)}$.

The function $\psi^{(1)}$ obeys the Dirac equation \eqref{Dirac} at every point with $x_3>0$, while 
\be\label{ex1psi0}
i\hbar \frac{\partial \psi^{(0)}}{\partial t} = \int\limits_{\RRR^2} dx_1\, dx_2\, N(x_1,x_2)^\dagger\, \psi^{(1)}(x_1,x_2,0)\,,
\ee
where $N(x_1,x_2)$ is a fixed spinor field that is square-integrable, 
\be\label{Nsqint}
\int\limits_{\RRR^2} dx_1\, dx_2\, N(x_1,x_2)^\dagger\, N(x_1,x_2) <\infty\,,
\ee
and satisfies
\be\label{Nalpha}
N^\dagger(x_1,x_2) \, \alpha^3 \, N(x_1,x_2) =0
\ee
at every $(x_1,x_2)\in\RRR^2$. For example, we can take it to be
\be\label{ex1N}
N(x_1,x_2) = e^{-x_1^2-x_2^2} \begin{pmatrix}1\\0\\1\\0\end{pmatrix}
\ee
expressed in the Weyl representation, in which \cite{gamma}
\be\label{Weyl}
\gamma^0 = \begin{pmatrix}0&I_2\\I_2&0\end{pmatrix}, \quad
\gamma^k = \begin{pmatrix} 0&\sigma^k\\-\sigma^k & 0\end{pmatrix}, \quad
\alpha^3 = \begin{pmatrix} -1&&&\\&1&&\\&&1&\\&&&-1 \end{pmatrix}
\ee
($k=1,2,3$) with $I_2$ the $2\times 2$ identity matrix and $\sigma^k$ the three Pauli matrices. Equations \eqref{Dirac} and \eqref{ex1psi0} are supplemented by the IBC
\be\label{ex1ibc}
(\gamma^3-i) \psi^{(1)}(x_1,x_2,0)
=-\tfrac{i}{\hbar} (\gamma^3-i) \alpha^3 N(x_1,x_2) \, \psi^{(0)}\,,
\ee
where $i$ is short for $i$ times the unit matrix. Note that $\gamma^3$ has eigenvalues $\pm i$, each eigenspace has dimension 2, and the two eigenspaces are orthogonal to each other (as $i\gamma^3$ is self-adjoint); as a consequence, $\gamma^3-i$ is $-2i$ times the projection to the eigenspace with eigenvalue $-i$, and \eqref{ex1ibc} constrains only two components of $\psi^{(1)}(x_1,x_2,0)$. In other words, although \eqref{ex1ibc} is an equation between 4-spinors, it amounts to only 2 independent complex equations per boundary point $(x_1,x_2,0)$. Note further that $(\gamma^3-i)N(x_1,x_2)$ automatically lies in the same eigenspace, so that \eqref{ex1ibc} can always be satisfied. The abstract structure of the IBC \eqref{ex1ibc} thus is to require, for every $(x_1,x_2)$, the pair $\bigl(\psi^{(0)},\psi^{(1)}(x_1,x_2,0)\bigr)$ to lie in a particular 3-dimensional subspace $\mathscr{L}(x_1,x_2)$ of $\CCC\oplus \CCC^4$. 

The model is not translation invariant in the $x_1$ and $x_2$ directions, and it is easy to understand why it cannot be: If the time evolution were $x_1x_2$-translation invariant, then consider the initial wave function $\psi^{(0)}=1$ and $\psi^{(1)}(\vx)=0$, which is also $x_1x_2$-translation invariant. It would follow that $\psi_t$ at any time $t$ has to be $x_1x_2$-translation invariant, but that can only be square-integrable if $\psi_t^{(1)}=0$, that is, if no creation occurs. The violation of $x_1x_2$-translation invariance is also connected to the square-integrability of the $N$ function, which is needed for the self-adjointness of $H$ (which again is closely related to the square-integrability of $\psi_t$ for all $t$).

\subsection{Conservation of Probability}
\label{sec:ex1cons}

As in the non-relativistic case \cite{TT15a,TT15b}, the IBC and the term on the right-hand side of \eqref{ex1psi0} are chosen so as to allow an exchange of $|\psi|^2$ probability between the two sectors while conserving the total amount of $|\psi|^2$. This is verified by the following calculation.

The probability current for the Dirac equation \eqref{Dirac} is
\be\label{vjdef}
\vj = (j^1,j^2,j^3) = \psi^\dagger \valpha \psi\,,
\ee
which comprises the 3 spacelike components of the 4-current
\be\label{jmudef}
j^\mu = \overline{\psi}\gamma^\mu \psi\,.
\ee
In our case, $\vj=\psi^{(1)\dagger} \valpha \psi^{(1)}$. The amount of probability lost per unit time in $\Q^{(1)}$ (which could be positive or negative) is given by the probability current into the boundary,
\begin{align}
\text{loss} ^{(1)}
&= -\int\limits_{\RRR^2} dx_1\, dx_2\,  j^3(x_1,x_2,0)\\
&=-\int\limits_{\RRR^2} dx_1\, dx_2\,  \psi^{(1)}(x_1,x_2,0)^\dagger \alpha^3 \psi^{(1)}(x_1,x_2,0)
\intertext{(the minus sign arising from the fact that $j^3$ is the current in the \emph{positive} $x_3$ direction). The amount of probability gained per unit time in $\Q^{(0)}$ (again, positive or negative) is, by \eqref{ex1psi0},}
\text{gain}^{(0)}
&= \frac{\partial |\psi^{(0)}|^2}{\partial t}\\
&= \tfrac{2}{\hbar} \Im (\psi^{(0)*} H\psi^{(0)})\\
&= \tfrac{2}{\hbar} \Im \Bigl[ \psi^{(0)*} \int\limits_{\RRR^2} dx_1\, dx_2\, N(x_1,x_2)^\dagger\, \psi^{(1)}(x_1,x_2,0) \Bigr]\,.
\end{align}
In order to check that gain$^{(0)}$ = loss$^{(1)}$, it suffices to show that the IBC \eqref{ex1ibc} implies that for every $x_1$ and $x_2$,
\be\label{check1}
\tfrac{2}{\hbar} \Im \Bigl[ \psi^{(0)*} N(x_1,x_2)^\dagger\, \psi^{(1)}(x_1,x_2,0) \Bigr]
=- \psi^{(1)}(x_1,x_2,0)^\dagger \alpha^3 \psi^{(1)}(x_1,x_2,0)\,.
\ee
Writing $\psi_1^{(1)},\psi_2^{(1)},\psi_3^{(1)},\psi_4^{(1)}$ for the four components of $\psi^{(1)}(x_1,x_2,0)$ in the Weyl representation (and $N_1,N_2,N_3,N_4$ for the components of $N$), the IBC \eqref{ex1ibc} can equivalently be rewritten as the conjunction of the two equations
\begin{subequations}
\begin{align}
\psi_1^{(1)} &= -i\psi_3^{(1)} + \tfrac{i}{\hbar} (N_1-iN_3) \psi^{(0)}\\
\psi_4^{(1)} &= -i\psi_2^{(1)} + \tfrac{1}{\hbar} (N_2+iN_4) \psi^{(0)}\,.
\end{align}
\end{subequations}
They allow us to eliminate $\psi_1^{(1)}$ and $\psi_4^{(1)}$ from the expressions in \eqref{check1}, leading to
\begin{align}
\text{lhs\eqref{check1}} &= \tfrac{2}{\hbar}\Re \Bigl[\psi^{(0)*} (-N_1^*-iN_3^*)\psi_3^{(1)} + \psi^{(0)*} (-iN_2^*-N_4^*) \psi_2^{(1)} \Bigr] \nonumber\\
&\quad + \tfrac{2}{\hbar^2} \Bigl( |N_1|^2+|N_4|^2 \Bigr) |\psi^{(0)}|^2 \nonumber\\
&\quad +\tfrac{2}{\hbar^2} \Re \Bigl[-iN_1^* N_3-iN_4^*N_2 \Bigr] |\psi^{(0)}|^2
\end{align}
and
\begin{align}
\text{rhs\eqref{check1}} &= \tfrac{2}{\hbar} \Re \Bigl[ \psi_3^{(1)*} (-N_1+iN_3) \psi^{(0)} + \psi_2^{(1)*} (iN_2 -N_4) \psi^{(0)} \Bigr]\nonumber \\
&\quad + \tfrac{1}{\hbar^2}\Bigl(|N_1|^2+|N_2|^2+|N_3|^2+|N_4|^2 \Bigr)|\psi^{(0)}|^2 \nonumber\\
&\quad +\tfrac{2}{\hbar^2} \Re\Bigl[ -iN_1^* N_3 +i N_2^* N_4 \Bigr] |\psi^{(0)}|^2\,.
\end{align}
Thus, using \eqref{Weyl},
\begin{align}
\text{lhs\eqref{check1}}-\text{rhs\eqref{check1}} 
&= \tfrac{1}{\hbar^2} \Bigl( |N_1|^2-|N_2|^2-|N_3|^2+|N_4|^2 \Bigr) |\psi^{(0)}|^2\\
&= -\tfrac{1}{\hbar^2} N^\dagger \alpha^3 N\, |\psi^{(0)}|^2\,,
\end{align}
which vanishes by virtue of \eqref{Nalpha}. This completes our derivation of \eqref{check1} and therefore our non-rigorous check of probability conservation. A rigorous result will be presented in Section~\ref{sec:sa}.

\subsection{Reflecting Boundary Condition}

When we set $N=0$, the Schr\"odinger equation \eqref{ex1psi0} for $\psi^{(0)}$ reduces to
\be
i\hbar \frac{\partial \psi^{(0)}}{\partial t} =0\,,
\ee
and the IBC \eqref{ex1ibc} reduces to
\be\label{bc3}
(\gamma^3-i) \psi^{(1)}(x_1,x_2,0) = 0\,.
\ee
In this case, loss$^{(1)}$ = 0, so there is no current into the boundary, and $\|\psi^{(0)}\|^2$ and $\|\psi^{(1)}\|^2$ are separately conserved. Thus, \eqref{bc3} is a \emph{reflecting} boundary condition. Indeed, \eqref{bc3} is used in the MIT bag model, and is known to make the Dirac Hamiltonian on the upper half space $\RRR^3_{>}$ self-adjoint \cite{FR15}.

\newpage

\noindent{\bf Remark.}
\begin{enumerate}
\setcounter{enumi}{\theremarks}
\item The condition \eqref{bc3} requires $\psi$ to lie, at every boundary point, in the 2-dimensional eigenspace of $\gamma^3$ with eigenvalue $i$; put differently, it specifies two of the four components of $\psi$ (in a suitable basis) on the boundary. The need for specifying just two of the components can be understood as follows. The boundary condition needs to specify a law for how any wave arriving at the boundary surface gets reflected. To this end, consider a wave function of the form
\be\label{reflected}
\psi(\vx) = u\, e^{i\vk'\cdot \vx} + v \, e^{i\vk\cdot \vx}
\ee
with $u,v\in \CCC^4$, $\vk=(k_1,k_2,k_3)\in \RRR^3$ with $k_3>0$, and $\vk'=(k_1,k_2,-k_3)$ (note the different sign). The first term in \eqref{reflected} is the incoming wave, the second the reflected wave. A law for reflection needs to determine the amplitude $v$ of the reflected wave for any given $u$, where the only incoming waves we need to consider are eigenwaves of the Dirac equation, so that $u$ is subject to the condition
\be
Eu = (\hbar \vk' \cdot\valpha + m\beta) u
\ee
with $E=\sqrt{m^2+\hbar^2\vk^{\prime 2}}$ the energy of the incoming wave. The matrix $\hbar \vk'\cdot \valpha + m \beta$ has eigenvalues $\pm E$, and each eigenspace $\mathscr{E}_{\pm E}(\vk')$ has dimension 2. Thus, $u$ has to lie in the 2-dimensional subspace $\mathscr{E}_{+E}(\vk')$ of $\CCC^4$, and correspondingly, $v$ has to lie in the 2-dimensional subspace $\mathscr{E}_{+E}(\vk)$, with the consequence that two equations are needed for determining $v$.  The boundary condition \eqref{bc3} amounts to
\be
(\gamma^3-i)(u+v)=0\,,
\ee
and this equation indeed provides two equations that determine $v$ from $u$. The wave function $\psi$ of \eqref{reflected} then is an eigenfunction of the Hamiltonian (where the boundary condition defines the domain of the Hamiltonian), so another way of viewing the need to determine $v$ from $u$ is that it comes from the need to determine the eigenfunctions of the Hamiltonian.
\end{enumerate}
\setcounter{remarks}{\theenumi}

\subsection{Self-Adjoint Hamiltonian}
\label{sec:sa}

Here is a rigorous result about the question whether the Hamiltonian $H$ is actually well defined and self-adjoint. The rigorous definition of $H$ includes the IBC in the choice of the \emph{domain} of $H$.

Instead of the half space $\RRR^3_>$, we consider a general region $\Omega\subset \RRR^3$ with $C^2$ boundary $\partial \Omega$. For technical reasons (due to the results from the literature we are using), we will assume that $\partial \Omega$ is 
compact. This assumption excludes the case of the half space $\Omega=\RRR^3_>$ discussed in Section~\ref{sec:defex1} but includes the cases in which $\Omega$ is the ball $B_r(\vzero)$ of radius $r>0$ around the origin $\vzero$ (as in the MIT bag model) or $\Omega=\RRR^3\setminus B_r(\vzero)$ (as for the creation and annihilation of point particles by a source that is a sphere of radius $r$ around the origin).

For all $\vx\in \partial \Omega$ let $\vn(\vx)$ denote the inward (i.e., toward $\Omega$) unit normal vector on the surface $\partial \Omega$ and $d^2\vx$ the surface area measure on $\partial \Omega$. We write
\be
\alpha^{\vn}(\vx) := \vn(\vx) \cdot \valpha ~~~\text{and}~~~
\gamma^{\vn}(\vx) := \vn(\vx) \cdot (\gamma^1,\gamma^2,\gamma^3)\,.
\ee
Let $N:\partial\Omega\to \CCC^4$ be a spinor field that satisfies
\be\label{Nalphaex2}
\int_{\partial \Omega} d^2\vx \, N(\vx)^\dagger \, \alpha^{\vn}(\vx) \, N(\vx) =0\,.
\ee
This assumption is weaker than the analog of \eqref{Nalpha}, i.e.,
\be\label{Nalpha3}
N(\vx)^\dagger \, \alpha^{\vn}(\vx) \, N(\vx) =0~~~~~\forall\vx\in \partial \Omega\,.
\ee
While \eqref{Nalpha3} is needed for the \emph{local} conservation of probability as in \eqref{check1}, \eqref{Nalphaex2} will be sufficient for \emph{global} conservation of probability and thus for self-adjointness.
 
The Hilbert space is $\Hilbert=\Hilbert^{(0)} \oplus \Hilbert^{(1)} := \CCC \oplus \Lz(\Omega,\CCC^4)$, the Schr\"odinger equation is given by the Dirac equation \eqref{Dirac} for $\psi^{(1)}$ in the interior of $\Omega$ and by
\be\label{ex2psi0}
i\hbar \frac{\partial \psi^{(0)}}{\partial t} = \int\limits_{\partial\Omega} d^2\vx\, N(\vx)^\dagger\, \psi^{(1)}(\vx)
\ee
for $\psi^{(0)}$.

\begin{thm}\label{thm:ex2}
Suppose that $\Omega\subset\RRR^3$ is open with compact $C^2$ boundary, and that $N\in H^{1/2}(\partial\Omega,\CCC^4)$ (i.e., the Sobolev space of degree $1/2$) satisfies \eqref{Nalphaex2}. Then the following operator $H$ in the Hilbert space $\Hilbert=\CCC\oplus \Lz(\Omega,\CCC^4)$ is well defined and self-adjoint: The domain $\domain$ of $H$ consists of those elements of $\CCC\oplus H^1(\Omega,\CCC^4)$ satisfying the IBC
\be\label{ex2ibc}
(\gamma^{\vn}(\vx)-i) \psi^{(1)}(\vx)
=-\tfrac{i}{\hbar} (\gamma^{\vn}(\vx)-i) \alpha^{\vn}(\vx)\, N(\vx) \, \psi^{(0)}\,,
\ee
and for $\psi\in \domain$, $(H\psi)^{(0)}$ and $(H\psi)^{(1)}$ are given by the right-hand sides of \eqref{ex2psi0} and \eqref{Dirac}, respectively.
\end{thm}

We give the proof in Appendix~\ref{app:sa}.

\section{Reflecting Boundary Conditions for Dirac-Type Hamiltonians}
\label{sec:reflect1}

Before we turn to general interior-boundary conditions in Section~\ref{sec:general}, we need a discussion of general reflecting boundary conditions.

\subsection{Setup: Configuration Space, Hilbert Space, and Dirac-Type Differential Operators}
\label{sec:config1}

We take the configuration space $\Q$ to be a manifold with boundary.\footnote{According to the definition of a manifold with boundary, every interior point has a neighborhood on which a coordinate chart is defined whose image is an open set in $\RRR^d$ for some $d$, while every boundary point has a neighborhood on which a coordinate chart is defined whose image is the intersection of an open set and a closed half-space in $\RRR^d$. In particular, the boundary has codimension 1, i.e., dimension $d-1$. The boundary may be empty.}
We write $\partial \Q$ for the boundary and $\Q^\circ=\Q\setminus \partial \Q$ for the interior of $\Q$. We take $\Q$ to be equipped with a Riemann metric $g_{ab}$, which also defines a volume measure $\mu$ on $\Q$; likewise, the metric defines a surface area measure $\lambda$ on $\partial \Q$. 

The wave function $\psi$ is a spinor-valued function on $\Q$, $\psi:\Q\to \CCC^{r}$; we denote the inner product in spin space $\CCC^{r}$ by $(\psi | \phi) = \psi^\dagger \phi$. More generally, we can take $\psi$ to be a cross-section of a vector bundle $E$ over $\Q$ of finite rank $r=\dim_{\CCC} E_q$ (dimension of fiber spaces). This case comes up, for example, when considering the Dirac equation in curved space-time (see, e.g., \cite{FR15}). We assume that $E$ is a \emph{Hermitian vector bundle}, i.e., a complex vector bundle equipped with a positive definite Hermitian inner product $(\ |\ )_q$ in every fiber $E_q$ and a metric connection, i.e., a connection relative to which the inner product is parallel, or
\be\label{Hermitianbundle}
\nabla\Bigl(\psi(q) \Big| \phi(q) \Bigr)_q=\Bigl(\nabla\psi(q) \Big| \phi(q) \Bigr)_q + \Bigl(\psi(q) \Big| \nabla\phi(q) \Bigr)_q\,.
\ee
We also write $|\psi(q)|^2$ for $\bigl(\psi(q) \big| \psi(q)\bigr)_q$, which is the density relative to $\mu$ of the probability distribution in $\Q$ associated with $\psi\in\Hilbert$ with $\|\psi\|=1$.

The Hilbert space $\Hilbert= \Lz(\Q,E,\mu)$ consists of the square-integrable cross-sections of $E$ and is equipped with the inner product
\be\label{inprdef}
\scp{\psi}{\phi} = \int_{\Q} \mu(dq) \, \bigl(\psi(q) \big| \phi(q)\bigr)_q\,.
\ee

A \emph{Dirac-type operator} is a differential expression of first order,\footnote{\label{fn:Clifford}When defining a ``Dirac-type operator,'' some authors (e.g., \cite{BB13}) demand that the coefficients $A^a(q)$ satisfy the Clifford relations, $A^a A^b +A^b A^a = 2g^{ab}$. However, while that is true of single-particle Dirac Hamiltonians, also in curved space-time and in any dimension, it is not true of many-particle Dirac Hamiltonians. For example, for two Dirac particles in Euclidean 3-space, the configuration space $\Q$ can be taken to be the orthogonal sum of two Euclidean 3-spaces ($d_n=6$) with $A^1=\alpha^1_1$ (i.e., $\alpha^1$ acting on the first spin index), $A^2=\alpha^2_1, A^3=\alpha^3_1, A^4=\alpha^1_2,A^5=\alpha^2_2, A^6=\alpha^3_2$, so $A^1$ and $A^4$ commute (as they act on different indices) instead of anti-commute (as would correspond to Clifford relations). At any rate, we will not use Clifford relations in the following.}
\be\label{Hexpression}
H\psi(q) = -i\hbar\sum_{a=1}^{d} A^a(q) \nabla_{\!a}\psi(q) + B(q)\psi(q)\,,
\ee
where $d=\dim_{\RRR}\Q$, $\nabla$ is the covariant derivative (corresponding to the connection of $E$), and $A^a(q)$ and $B(q)$ are endomorphisms 
of $E_q$ or, equivalently, elements of $E_q\otimes E_q^*$, where the star denotes the dual space; more precisely, $A(q)\in E_q\otimes E_q^*\otimes \CCC T_q \Q$. 
The choice of the arbitrary prefactor as $-i\hbar$ will be convenient later. By a \emph{Dirac-type equation} we mean the associated Schr\"odinger equation
\be\label{Schr2}
i\hbar\partial_t \psi = H\psi\,.
\ee 
In order to define a time evolution, the expression \eqref{Hexpression} will need to be supplemented by boundary conditions that we discuss below. 

\bigskip

\noindent{\bf Example~1.} The free Dirac equation \eqref{Dirac} for a single particle in flat space-time corresponds to $\Q$ being Euclidean 3-space, the boundary $\partial\Q$ being empty, and $E$ being the trivial vector bundle $\Q\times \CCC^4$ (i.e., $E_q=\CCC^4$) equipped with the standard inner product on $\CCC^4$, $(\psi | \phi) = \psi^\dagger\phi=\overline{\psi}\gamma^0\phi$, and the trivial connection (so that covariant derivatives coincide with partial derivatives); $\Hilbert=\Lz(\RRR^3,\CCC^4)$; $A^a(q)=\alpha^a$ ($a=1,2,3$) are the Dirac alpha matrices, and $B(q)=m\beta$ with $\beta$ the Dirac beta matrix.

For $N>1$ identical free Dirac particles, $\Q$ can be taken to be $(\RRR^3)^N$, $\partial \Q=\emptyset$, $E=\Q\times (\CCC^4)^{\otimes N}$ with the trivial connection, $(\psi | \phi) = \psi^\dagger \phi = \overline{\psi} \gamma^0\otimes \cdots \otimes \gamma^0 \phi$, $\Hilbert$ comprises the anti-symmetric functions in $\Lz((\RRR^3)^N,(\CCC^4)^{\otimes N})$, $A^{3n+i-3}(q) = \alpha_n^i$ ($i=1,2,3$ and $n=1,\ldots,N$), and $B(q) = m\sum_{n=1}^N \beta_n$. (Alternatively, we can take $\Q$ to be the space of unordered configurations \cite{fermionic} $\{q\subset \RRR^3: \# q=N\}$, and $E$ the tensor product of the fermionic line bundle and the set-indexed tensor product of spin spaces $\CCC^4$.)$\hfill\square$

\bigskip

We want to obtain a continuity equation
\be\label{conti}
\partial_t |\psi(q)|^2 = - \mathrm{div}\, j(q) := -\sum_{a=1}^{d} \nabla_{\!a} j^a(q)\,,
\ee
where $j$ is a (time-dependent) vector field on $\Q$ playing the role of the probability current and $\nabla$ is the covariant derivative defined by the Riemann metric on $\Q$. (The notation $\nabla_{\!a} j^a$ means, as in general relativity, that we first take the derivative of the vector field $j$ and then take the $aa$-component of the result, rather than the derivative of the scalar function that is the $a$-th component of the vector field $j$. That makes a difference in the case of curved metrics.)
From \eqref{Hexpression} and \eqref{Schr2}, we do obtain that
\begin{align}
\partial_t |\psi(q)|^2
&= \tfrac{2}{\hbar}\, \Im \Bigl(\psi(q) \Big| H\psi(q) \Bigr)_{\!q} \\
&= -\sum_{a=1}^{d} \biggl(\psi(q) \bigg|  A^a(q) \nabla_{\!a}\psi(q)\biggr)_{\!\!q}  -\sum_{a=1}^{d} \biggl(\nabla_{\!a}\psi(q) \bigg| A^a(q)^\dagger \,\psi(q)\biggr)_{\!\!q} \nonumber\\
&\quad +\: \tfrac{2}{\hbar}\, \Im \biggr(\psi(q) \bigg|  B(q)\psi(q)\biggr)_{\!\!q} \quad,\label{last}
\end{align}
where the dagger ${}^\dagger$ denotes the adjoint endomorphism relative to $(\ |\ )_q$. The expression \eqref{last} is of the form \eqref{conti} with
\be\label{jdef}
j^a(q) = \Bigl(\psi(q) \Big| A^a(q)\psi(q)\Bigr)_{\!q} \,,
\ee
provided that
\be\label{Asa}
A^a(q) \text{ is self-adjoint}
\ee
and $B(q)$ is of the form
\be\label{Bsa}
B(q) = B_0(q) -\tfrac{i\hbar}{2}\sum_{a=1}^{d} \nabla_{\!a} A^a(q) 
\text{ with self-adjoint $B_0(q)$.} 
\ee
We henceforth assume that the conditions \eqref{Asa} and \eqref{Bsa} are fulfilled.

\bigskip

\noindent{\bf Example~2.} For the free Dirac equation, \eqref{Asa} is satisfied because the Dirac alpha matrices are self-adjoint, and \eqref{Bsa} is satisfied because the Dirac beta matrix is self-adjoint and $\nabla_{\!a}A^a(q)=0$ (because $A^a(q)=\alpha^a$ is constant and $\nabla_{\!a}=\partial_a$). Note also that the general definition \eqref{jdef} of the current $j$ agrees with the earlier specific one in \eqref{vjdef}.$\hfill\square$

\bigskip

Conditions \eqref{Asa} and \eqref{Bsa} are the formal (algebraic) conditions needed for self-adjoint\-ness. It is known \cite{Che} that if $\Q$ has no boundary and is complete,\footnote{``Completeness'' requires that boundary points cannot be left out of $\Q$ but must be included in the manifold-with-boundary.} and if the propagation speed $c(q)$ (see below) is bounded on $\Q$, then $H$ given by \eqref{Hexpression} extends uniquely (from $C_c^\infty(E)$, the space of smooth compactly supported cross-sections) to a self-adjoint operator in $\Lz(E)$. The propagation speed $c(q)$ of wave functions at $q$ is the supremum over $u\in T_q \Q$ with $|u|=1$ of the operator norm (largest absolute eigenvalue) of $u_a A^a(q)$.

\subsection{Boundary Conditions}
\label{sec:known}

Most known boundary conditions (e.g., \cite{BB13,FR15}) are \emph{reflecting} boundary conditions that will make the Hamiltonian self-adjoint while involving no interior points (but see also \cite{detect-dirac} for \emph{absorbing} boundary conditions). Before setting up IBCs, it will be useful to recap the general form of reflecting boundary conditions. 

Consider a boundary point $q\in\partial \Q$, let $n(q)$ denote the inward-pointing\footnote{Inward-pointing means that there is a curve $c:[0,\delta)\to \Q$ with $c(0)=q$ and $\dot{c}(q)=n(q)$.} unit normal vector to the boundary at $q$ (relative to the Riemann metric $g_{ab}$), and let $A^n$ be the endomorphism of $E_q$ given by
\be\label{Andef}
A^n= n(q)\cdot A(q) = \sum_{a,b=1}^{d_n}n^a(q)\, A^b(q)\, g_{ab}(q)\,. 
\ee
Let $E^0_q$ denote the kernel of $A^n$ and $E^{\pm}_q$ the sum of the eigenspaces with positive (negative) eigenvalues, so
\be
E_q = E^+_q \oplus E^-_q \oplus E^0_q\,.
\ee
Let $P^0$ denote the orthogonal projection to $E^0_q$ and $P^{\pm}$ that to $E^{\pm}_q$; we write $A^{\pm}$ for the restriction of $A^n(q)$ to $E^{\pm}_q$. Put differently, $P^0=1_{\{0\}}(A)$, $P^+=1_{(0,\infty)}(A)$, and $P^-=1_{(-\infty,0)}(A)$, where $1_S(x)$ means the characteristic function of the set $S$.

Let us derive what reflecting linear boundary conditions look like. Such a condition must exclude any current into the boundary (to enable $H$ as in \eqref{Hexpression} to be self-adjoint); that is, it must ensure that
\be
j^n(q) = n(q) \cdot j(q) =0
\ee
at every boundary point $q$. Thus, it must specify a subspace $S_q \subset E_q$ such that
\be\label{psiAnpsi=0}
\bigl( \psi\big| A^n \, \psi \bigr)_q=0 
\ee
for every $\psi\in S_q$. Splitting $\psi$ into its parts in $E_q^0,E_q^+,E_q^-$, we can rewrite \eqref{psiAnpsi=0} as
\be
\bigl( P^+\psi\big| A^+ P^+\psi \bigr)_q
= \bigl( P^-\psi\big| -A^- P^-\psi \bigr)_q
\ee
or
\be
\Bigl\| (A^+)^{1/2} P^+\psi \Bigr\|^2
= \Bigl\| (-A^-)^{1/2} P^-\psi \Bigr\|^2 \,.
\ee
This will follow if we specify a unitary isomorphism $L: E^+_q \to E^-_q$ and demand that
\be\label{bc1}
(-A^-)^{1/2} P^-\psi = L \, (A^+)^{1/2} P^+\psi \,,
\ee
a linear condition that fixes $P^-\psi$ in terms of $P^+ \psi$. This, with $L$ depending on $q$, is the desired reflecting boundary condition, and $S_q$ is the set of $\psi\in E_q$ satisfying \eqref{bc1}. In Section~\ref{sec:Lagrangian} we will enter a deeper analysis and argue that these are the only possible reflecting boundary conditions. 

\bigskip

\noindent{\bf Example~3.} For the free Dirac equation on the upper half space $\RRR^3_{>}$,
we have that $A^n=\alpha^3$, which has eigenvalues $\pm 1$, each with multiplicity 2. Thus, $\dim E_q^+=\dim E_q^-$, both $A^+$ and $-A^-$ are the identity on their respective domains, and the boundary condition \eqref{bc1} reduces to
\be\label{bc2}
\Bigl[ P^- -L(x_1,x_2) P^+ \Bigr] \psi(x_1,x_2,0)=0
\ee
with $L(x_1,x_2)$ a unitary isomorphism between the two eigenspaces. In particular, the boundary condition indeed specifies two of the four components of $\psi$ on the boundary (viz., $P^- \psi$). In the Weyl representation \eqref{Weyl}, in which
$\alpha^3 = \diag(-1,1,1,-1)$, we have that $P^0=0$, $P^+=\diag(0,1,1,0)$, and $P^-=\diag(1,0,0,1)$. 
A particular choice of $L: E^+_q\to E^-_q$ is (expressed as a matrix acting on the whole spin space $\CCC^4$)
\be\label{ex1L}
LP^+ = \begin{pmatrix} 0&0&-i&0\\0&0&0&0\\0&0&0&0\\0&-i&0&0 \end{pmatrix},~~\text{so}~~
R=\begin{pmatrix} 1&0&i&0\\0&0&0&0\\0&0&0&0\\0&i&0&1 \end{pmatrix}
\ee
independently of $q$; for this choice, the boundary condition \eqref{bc2} becomes
\be
\psi_1 = -i\psi_3\,, \qquad \psi_4 = -i\psi_2 \qquad \text{on }\{x^3=0\}
\ee
in the Weyl representation, which is equivalent to $(\gamma^3-i)\psi=0$
as in \eqref{bc3}.$\hfill\square$

\newpage

\noindent{\bf Remark.}
\begin{enumerate}
\setcounter{enumi}{\theremarks}

\item It is easy to understand why any boundary condition of the form 
\be\label{bc-+}
P^-\psi = C P^+\psi
\ee
with some linear mapping $C:E_q^+\to E_q^-$ determining $P^-\psi$ from $P^+\psi$ would lead to \eqref{bc1}: It would yield
\be
j^n(q) = \bigl( P^+\psi \big| A^+ P^+ \psi \bigr) + \bigl( CP^+\psi  \big| A^- CP^+\psi  \bigr)\,,
\ee
which vanishes for all choices of $P^+\psi$ if and only if
\be\label{ACAC}
A^+ + C^\dagger A^- C =0\,.
\ee
If we define
\be\label{LAC}
L= (-A^-)^{1/2} C (A^+)^{-1/2}\,, 
\ee
then 
\be\label{CAL}
C= (-A^-)^{-1/2} L (A^+)^{1/2}\,,
\ee 
and \eqref{ACAC} is equivalent to
\be
I_{E_q^+} =  L^\dagger L\,,
\ee
which means that $L$ must be unitary to its image.
\end{enumerate}
\setcounter{remarks}{\theenumi}

\subsection{Complete Lagrangian Subspaces}
\label{sec:Lagrangian}

Before we can enter the discussion of the general IBC, we need a more thorough discussion of the possible reflecting boundary conditions. We want to identify the subspaces $S_q$ that can be used as a reflecting boundary condition. Let $\Kegel$ denote the set of all $\psi\in E_q$ satisfying \eqref{psiAnpsi=0}; it is not a subspace but a cone. We demand that $S_q \subseteq \Kegel$, but that will not be sufficient; for example, it is well known (e.g., \cite{ADV97a}) that $S_q=\{0\}$ (i.e., the homogeneous Dirichlet boundary condition $\psi|_{\partial \Q}=0$) is \emph{not} a possible boundary condition for the Dirac equation.

Here is what we argue is the right condition. Let us begin with a few definitions: 
For any subspace $S$ of $E_q$, let 
\be
S^\#=\bigl\{\phi\in E_q:  (\phi|A^n \chi)=0~~\forall \chi \in S \bigr\}\,.
\ee
A subspace $S\subseteq E_q$ is called a \emph{complete Lagrangian subspace} relative to $A^n$ iff\footnote{iff = if and only if} 
\be
S=S^\#\,.
\ee
Equivalently, $S$ is complete Lagrangian iff
\be
\phi \in S ~\Leftrightarrow~ \forall\chi\in S: \bigl( \phi\big| A^n \, \chi \bigr)=0 \,.
\ee

\begin{conj}\label{conj:refl}
Let $S_q\subseteq E_q$ be a subbundle. The boundary condition
\be\label{bcS}
\psi(q) \in S_q~~\forall q\in \partial\Q
\ee
can occur in a self-adjoint extension (from $C_c^\infty(E|_{\Q^\circ})$) of $H$ in $\Lz(E)$ iff $S_q$ is a complete Langrangian subspace of $E_q$ relative to $A^n$ for every $q\in\partial\Q$.
\end{conj} 

Here is why this is plausible. For the Hamiltonian $H$ with boundary condition \eqref{bcS}, integration by parts yields that
\be
\scp{\phi}{H\psi} - \scp{H\phi}{\psi} = i\hbar \int_{\partial \Q} \!\!\! \lambda(dq) \, \bigl( \phi(q) \big| A^n \psi(q)\bigr)_q \,.
\ee 
First of all, this needs to vanish for all $\phi$ and $\psi$ in the domain, and since we are interested in local boundary conditions, we need the integrand $\bigl( \phi(q) \big| A^n \psi(q)\bigr)_q$ to vanish pointwise, i.e., $S_q \subset \Kegel$. Second, we also need that the domain of the \emph{adjoint} of $H$-with-\eqref{bcS} is no greater than the domain of $H$-with-\eqref{bcS}; to this end, we need that any $\phi$ with the property that $\bigl( \phi(q) \big| A^n \psi(q)\bigr)_q=0$ for all $\psi$ satisfying the boundary condition, satisfies itself the boundary condition; and that amounts to the complete Lagrangian property of $S_q$. See \cite{EM05} for a broader discussion of the relevance of complete Lagrangian subspaces to self-adjoint extensions.

\begin{prop}\label{prop:refl}
Let $E$ be a finite-dimensional complex Hilbert space, $A:E\to E$ a self-adjoint endomorphism, $P^0=1_{\{0\}}(A)$, $P^+=1_{(0,\infty)}(A)$, $P^-=1_{(-\infty,0)}(A)$, and let $E^0$ and $E^\pm$ be the ranges of $P^0$ and $P^\pm$, respectively. The complete Lagrangian subspaces $S$ of $E$ relative to $A$ are in a natural bijective relation to the unitary isomorphisms $L:E^+\to E^-$, given by
\be\label{SL}
S= \bigl\{\psi\in E: (-A^-)^{1/2} P^-\psi = L (A^+)^{1/2} P^+ \psi \bigr\}\,.
\ee  
In particular, complete Langrangian subspaces exist iff
\be\label{dim+-}
\dim E^+ = \dim E^-\,.
\ee
\end{prop}

We give the proof in Appendix~\ref{app:proofs}.

In the following, we will use the abbreviation
\be\label{Rdef}
R=\sqrt{-A^-}P^- - L\sqrt{A^+}P^+= (P^- -LP^+)\sqrt{|A^n|}
\ee
and write \eqref{bc1} as
\be
R(q) \psi(q) =0\,.
\ee
$R(q)$ is an endomorphism of $E_q$ whose kernel is $S_q$.

\subsection{Examples Involving Reflecting Boundary Conditions}

\noindent{\bf Example~4.} Several authors \cite{FG87,ADV97b} have determined all self-adjoint extensions of the Dirac operator in 1 space dimension, $H=-i\hbar\alpha^1 \partial+m\beta$, on an interval $[0,1]$. For the 1d Dirac equation, spin space is 2-dimensional, and we can take $\gamma^0=\beta=\sigma_1$, $\gamma^1=\sigma_1 \sigma_3$, $\alpha^1=\gamma^0\gamma^1=\sigma_3$. We write the components of $\psi$ as $\psi_-$ and $\psi_+$ (opposite to the eigenvalue of $\alpha^1$). All of the self-adjoint extensions involve boundary conditions, some of them non-local as they relate $\psi(0)$ to $\psi(1)$; those that are local are of the form
\be
\psi_+(0) = e^{i\theta} \psi_-(0)\,, \quad \psi_+(1) = e^{i\varphi} \psi_-(1)
\ee
with fixed parameters $\theta,\varphi\in\RRR$. These conditions are special cases of the above scheme \eqref{bc1} for reflecting boundary conditions with $A^n=\alpha^1=\diag(1,-1)$, $P^+=\diag(1,0)$, $P^-=\diag(0,1)$,
\be
LP^+= \begin{pmatrix} 0&0 \\ e^{i\theta}&0 \end{pmatrix}\,, \quad \text{so }
R= \begin{pmatrix} 0&0 \\ -e^{i\theta}&1 \end{pmatrix}
\ee
and correspondingly with $\varphi$.$\hfill\square$

\bigskip

\noindent{\bf Example~5.} Al-Hashimi and Wiese \cite[Sec.~4.2]{AHW12} considered the 3d Dirac equation \eqref{Dirac} in a region $\Omega\subset \RRR^3$ with smooth boundary $\partial\Omega$ and aimed at formulating ``the most general perfectly reflecting boundary condition.'' They used the Dirac representation of spin space \cite{gamma}, in which
\be
\alpha^k = \begin{pmatrix} 0&\sigma_k \\ \sigma_k &0\end{pmatrix}\,,
\ee
and proposed a family of boundary conditions 
\be\label{AHWbc}
\begin{pmatrix}\psi_{D3}\\\psi_{D4} \end{pmatrix} = i\, \vn(\vx)\cdot \vsigma\, T \begin{pmatrix} \psi_{D1}\\\psi_{D2}\end{pmatrix}
\ee
for all $\vx\in \partial\Omega$, where $\psi_{Di}$ are the four components of $\psi$ in the Dirac representation, $\vn(\vx)$ is the inward unit normal vector to $\partial\Omega$ at $\vx$, and $T$ is an arbitrary but fixed self-adjoint $2\times 2$ matrix. They showed that the subspace comprising those $\psi\in\CCC^4$ that satisfy \eqref{AHWbc} lies in $\Kegel$, i.e., that \eqref{AHWbc} implies $\psi^\dagger (\vn(\vx)\cdot \valpha) \psi=0$. In fact, \eqref{AHWbc} is a special case of \eqref{bc1} with (expressed in the Weyl representation and assuming that $\vn(\vx)$ points in the $x_3$-direction)
\be\label{LT}
LP^+ = \frac{1}{\mathscr{N}} \begin{pmatrix} 0&2ib&(1-ia)(1+ic)-|b|^2&0 \\ 0&0&0&0 \\ 0&0&0&0 \\ 0&(1+ia)(1-ic)-|b|^2&-2ib^*&0 \end{pmatrix}
\ee
for
\be
T= \begin{pmatrix} a&b\\b^*&c \end{pmatrix} \quad \text{with $a,c\in \RRR$ and $b\in\CCC$.} 
\ee
Here, the denominator is
\be
\mathscr{N} = (1+ia)(1+ic) + |b|^2\,,
\ee
which is always non-zero, as one easily checks.

However, Al-Hashimi and Wiese missed some boundary conditions, as some unitary $L:E_q^+ \to E_q^-$ are not of the form \eqref{LT}, for example
\be\label{AHWex}
LP^+= 
\begin{pmatrix} 0&0&e^{i\kappa}&0\\0&0&0&0\\0&0&0&0\\0&-1&0&0 \end{pmatrix} \text{ or }
\begin{pmatrix} 0&0&-1&0\\0&0&0&0\\0&0&0&0\\0&e^{i\kappa}&0&0 \end{pmatrix} \text{ or }
\begin{pmatrix} 0&e^{i\kappa}&0&0\\0&0&0&0\\0&0&0&0\\0&0&e^{-i\kappa}&0 \end{pmatrix} 
\ee
for any $\kappa\in\RRR$. More generally, starting from the boundary condition \eqref{bc1} or \eqref{bc2}, which in the Weyl representation (again with $\vn(\vx)$ in the $x_3$-direction) reads
\be
\begin{pmatrix} \psi_1\\\psi_4 \end{pmatrix} = L 
\begin{pmatrix} \psi_2\\\psi_3 \end{pmatrix}\,,
\ee
and writing $L$ in the form
\be
L = \begin{pmatrix} \cos \zeta e^{i\eta} &\sin \zeta e^{i\kappa}\\ \sin \zeta e^{i\tau} & -\cos \zeta e^{i(\kappa+\tau-\eta)} \end{pmatrix}
\ee
with $\zeta\in [0,\tfrac{\pi}{2}]$ and $\kappa,\tau,\eta\in [0,2\pi)$, as any unitary $2\times 2$ matrix can be written, one finds that
\be
T=-\frac{i}{\mathscr{N}'} \begin{pmatrix} 
1+\sin \zeta(e^{i\tau}-e^{i\kappa})-e^{i(\kappa+\tau)} & -2\cos \zeta e^{i\eta}\\
2\cos \zeta e^{i(\kappa+\tau-\eta)} & 1-\sin\zeta (e^{i\tau}-e^{i\kappa})-e^{i(\kappa+\tau)}
\end{pmatrix}
\ee
exists and is self-adjoint as soon as the denominator
\be
\mathscr{N}' = 1+\sin \zeta (e^{i\tau}+e^{i\kappa})+e^{i(\kappa+\tau)}
\ee
does not vanish. However, it does vanish for the examples in \eqref{AHWex}. The set of subspaces parametrized by the self-adjoint $2\times 2$ matrices $T$ is a non-compact (and thus proper) subset of the compact set of subspaces $S_q$ parametrized by the unitary $2\times 2$ matrices $L$.
$\hfill\square$

\bigskip

\noindent{\bf Example~6.} Benguria et al.~\cite{Ben17} proved the self-adjointness of the 2d Dirac operator (with spin space $\CCC^2$) in a region $\Omega\subset \RRR^2$ with $C^2$-smooth boundary $\partial \Omega$ and boundary condition
\be\label{bcBen}
P_{\eta(q),q} \psi(q) =0 \quad \forall q\in \partial\Q\,,
\ee
where $P_{\eta(q),q}$ is the projection to a certain 1d subspace of the spin space $\CCC^2$, viz., the eigenspace with eigenvalue $-1$ of $(\cos \eta(q)) \, \vt(q)\cdot \vsigma+(\sin\eta(q)) \, \vN \cdot \vsigma$, where $\eta(q)$ is an arbitrary real parameter, $\vt(q)$ is the unit tangent vector to $\partial \Q$ at $q$, and $\vN$ is the unit normal in $\RRR^3$ to the plane containing $\Omega$. (Except that they did not prove self-adjointness for cases with $\cos \eta(q)=0$.) In a Weyl representation with $\vn(q) \cdot \vsigma= \diag(1,-1)$ (where $\vn$ is the unit normal in $\Omega$ on $\partial \Omega$), so $P^+=\diag(1,0)$ and $P^-=\diag(0,1)$, the boundary condition \eqref{bcBen} is equivalent to
\be
P^- \psi = \begin{pmatrix} 0&0\\ e^{i\eta(q)} & 0\end{pmatrix}  P^+\psi
\ee
in our notation and thus a special case of the general boundary condition \eqref{bc1}.
$\hfill\square$

\bigskip

\noindent{\bf Example~7.} Lienert \cite{L15a,L15c,LN15} considered point interaction between massless Dirac particles in 1 space dimension, implemented through a reflecting boundary condition. For simplicity, we focus on the case of two particles. The boundary is the diagonal $\partial \Q = \{(z,z):z\in\RRR\}$ in configuration space $\Q=\{(z_1,z_2)\in\RRR^2:z_1\leq z_2\}$. For two particles, the appropriate Dirac equation reads
\be
i\hbar \partial_t \psi(z_1,z_2) = \Bigl(-i\hbar \alpha_1^1 \partial_{z_1} -i\hbar\alpha_2^1 \partial_{z_2}\Bigr) \psi(z_1,z_2)\,,
\ee
and we write the components of the wave function $\psi:\Q\to \CCC^2\otimes \CCC^2$ as
\be
\psi= \begin{pmatrix} \psi_{--}\\\psi_{-+}\\\psi_{+-}\\\psi_{++} \end{pmatrix} \in \CCC^2 \otimes \CCC^2\,.
\ee
In this representation, $A^1(q)=\alpha^1_1=\diag(1,1,-1,-1)$ and $A^2(q)=\alpha^1_2=\diag(1,-1,1,-1)$ for every $q\in\Q$. Thus, for every $q\in \partial \Q$,
\be\label{L15An}
A^n=-\tfrac{1}{\sqrt{2}} A^1(q) + \tfrac{1}{\sqrt{2}} A^2(q) =\diag\Bigl(0,-\sqrt{2},+\sqrt{2},0\Bigr)\,, 
\ee
so $P^+=\diag(0,0,1,0)$, $P^-=\diag(0,1,0,0)$, and $P^0=(1,0,0,1)$. The unitary $L:E_q^+ \to E_q^-$ is just a phase, $L=e^{i\theta}$ or, expressed as a matrix on the whole spin space,
\be\label{L15LP+}
LP^+=\begin{pmatrix} 0&0&0&0\\0&0&e^{i\theta}&0\\0&0&0&0\\0&0&0&0 \end{pmatrix}\,,~~~\text{so }
R=2^{1/4}\begin{pmatrix} 0&0&0&0\\0&1&-e^{i\theta}&0\\0&0&0&0\\0&0&0&0 \end{pmatrix}\,.
\ee
Correspondingly, Lienert's boundary condition is only one equation:
\be\label{BCL}
\psi_{-+}(z,z) - e^{i\theta}\psi_{+-}(z,z) = 0\,.
\ee
\hfill$\square$

\section{General Form of IBCs for Dirac-Type Hamiltonians}
\label{sec:general}

We now turn to the question what IBCs can look like in general for Hamiltonians that are first-order differential operators with matrix-valued coefficients.

\subsection{Setup: Configuration Space, Hilbert Space, and Dirac-Type Differential Operators}
\label{sec:config2}

We take the configuration space $\Q$ to be a finite or countable union of disjoint spaces with boundary, $\Q=\cup_n \Q^{(n)}$. For the sake of simplicity, we may assume that each $\Q^{(n)}$ is a manifold with boundary, although that implies that the boundary $\partial \Q^{(n)}$ of $\Q^{(n)}$ is itself a manifold \emph{without} boundary. In many applications, one may want to allow that the boundary $\partial \Q^{(n)}$ has itself a boundary; for example, this situation arises when $\Q^{(n)}$ is an $n$-particle configuration space of the form $\sM^n$, $n=0,1,2,\ldots$, where $\sM$ is a manifold with boundary.

We write $\partial \Q$ for $\cup_n \partial \Q^{(n)}$ and $\Q^\circ=\Q\setminus \partial \Q$ for the interior of $\Q$. We take $\Q$ to be equipped with a Riemann metric $g_{ab}$, which also defines a volume measure $\mu^{(n)}$ on $\Q^{(n)}$, and thus a measure $\mu$ on $\Q$, $\mu(S)=\sum_n \mu^{(n)}(S\cap \Q^{(n)})$; likewise, the metric defines a surface area measure $\lambda$ on $\partial \Q$. 

The wave function $\psi$ is a spinor-valued function on $\Q$. That is, its restriction to the sector $\Q^{(n)}$ of $\Q$ is a function $\psi^{(n)}:\Q^{(n)}\to \CCC^{r_n}$; we denote the inner product in spin space $\CCC^{r_n}$ by $(\psi | \phi) = \psi^\dagger \phi$. More generally, we can take $\psi^{(n)}$ to be a cross-section of a vector bundle $E^{(n)}$ over $\Q^{(n)}$ of finite rank $r_n=\dim_{\CCC} E_q^{(n)}$ (dimension of fiber spaces). We write $E$ for the entire vector bundle $\cup_n E^{(n)}$ and $E_q$ for $E_q^{(n)}$ if $q\in\Q^{(n)}$. We assume that $E^{(n)}$ is a \emph{Hermitian vector bundle} as discussed around \eqref{Hermitianbundle}. 

The Hilbert space $\Hilbert= \Lz(\Q,E,\mu)$ consists of the square-integrable cross-sections of $E$ and is equipped with the inner product \eqref{inprdef} with the understanding that $\int_{\Q}$ means the same as $\sum_n \int_{\Q^{(n)}}$, so that $\Hilbert=\oplus_n \Lz(\Q^{(n)},E^{(n)},\mu^{(n)})$.

A \emph{Dirac-type operator} will again be a differential expression of first order as in \eqref{Hexpression} (with $d=d_n=\dim_{\RRR}\Q^{(n)}$ now depending on the sector).

The IBC will be so constructed that the amount of probability per time that flows out of the boundary at $q'\in\partial \Q$ gets added to $|\psi|^2$ at an interior point
\be
q=f(q')
\ee
in a different sector,\footnote{In terms of the Bohmian trajectories \cite{bohmibc}, whenever the Bohmian configuration $Q(t)$ reaches the boundary at $q'$, it jumps to $f(q')$; conversely, from an interior point $q$, the Bohmian configuration can spontaneously jump to a boundary point $q'$ with $f(q')=q$, and this happens at a rate given in \cite{bohmibc}.} $f:\partial\Q \to \Q^\circ$. Since many boundary points $q'$ can be mapped to the same interior point $q$, we will need to talk of the set of those $q'$, which will be denoted
\be
f^{-1}(q)= \bigl\{ q'\in\partial\Q: f(q')=q \bigr\}\,,
\ee
and a measure $\nu_q$ over $f^{-1}(q)$. The appropriate (unnormalized) ``uniform'' measure for our purpose is characterized by
\be\label{munulambda}
\int_{\Q} \mu(dq) \, \nu_q \bigl( S\cap f^{-1}(q) \bigr) =\lambda (S)
\ee
for any set $S\subseteq \partial\Q$ (see \cite{co1} for more detail).

Again, we write $A^n$ for the component $n(q) \cdot A(q)$ of $A(q)$ normal to the boundary as in \eqref{Andef}; $E^0$, $E^+$, and $E^-$, respectively, for the kernel, the positive spectral subspace, and the negative spectral subspace of $A^n$; and $P^0$, $P^+$, and $P^-$ for the corresponding projections.

\subsection{A Class of IBCs}
\label{sec:ibc}

Before describing the most general IBC in Section~\ref{sec:full}, we write down a simpler type of IBC.

The Hamiltonian acts at every interior point $q$ according to
\be\label{Hdef}
H\psi(q) = -i\hbar\sum_{a=1}^{d_n} A^a(q) \nabla_{\!a}\psi(q) + B(q)\psi(q)
+ \!\!\! \int\limits_{f^{-1}(q)} \!\!\! \nu_q(dq') \, N(q')^\dagger\, \psi(q')\,.
\ee
Here, $N(q')$ is a complex-linear mapping $E_{f(q')}\to E_{q'}$, and the adjoint $N^\dagger:E_{q'}\to E_{q}$ of a linear mapping $N:E_{q}\to E_{q'}$ is defined by
the relation
$(\psi | N^\dagger \phi)_{q} = (N\psi | \phi)_{q'}$
for all $\psi\in E_{q}$ and $\phi\in E_{q'}$.

The IBC reads: for every boundary point $q$,
\be\label{IBC}
R(q) \, \psi(q) = -\tfrac{i}{\hbar} \, R(q) \, A^{\mathrm{inv}}(q)\, N(q) \, \psi\bigl(f(q)\bigr)\,,
\ee
with $R=R(q)$ as in \eqref{Rdef} (based again on a unitary isomorphism $L=L(q):E_q^+\to E_q^-$) and $A^{\mathrm{inv}}$ the inverse of $A^n$ on the orthogonal complement of its kernel,
\be\label{Ainvdef}
A^{\mathrm{inv}}= P^+ (A^+)^{-1} P^+ + P^- (A^-)^{-1} P^-\,.
\ee
If 0 is not an eigenvalue of $A^n$, then $A^{\mathrm{inv}}=(A^n)^{-1}$.

The function $N$ is required to satisfy
\be\label{Ncond1}
P_q^0 \, N(q) = 0
\ee
and
\be\label{Ncond2}
N(q)^\dagger A^{\mathrm{inv}}(q)\, N(q) =0
\ee
at every $q\in\partial \Q$. The conservation of probability will be verified later for the more general IBC of Section~\ref{sec:full}.

\bigskip

\noindent{\bf Example~8.} In the example of Section~\ref{sec:ex1}, $f(q')=\emptyset\in \Q^{(0)}$ for every $q'\in \partial \Q^{(1)}$, so $f^{-1}(\emptyset)=\partial \Q^{(1)}$; furthermore, $\nu_{\emptyset}(dx_1 \times dx_2) = dx_1 \, dx_2$, $E_{\emptyset}=\CCC$, and $E_q=\CCC^4$ for $q\in \Q^{(1)}$ with $(\psi|\phi)_q = \psi^\dagger \phi$ as before. As in Example~3, $A^n=\alpha^3$ at every boundary point, so $P^0=0$, $P^+=\diag(0,1,1,0)$, and $P^-=\diag(1,0,0,1)$. Furthermore, $L$ and $R$ are again given by \eqref{ex1L}. 
The $N$ function is given by \eqref{ex1N}, which satisfies 
\eqref{Ncond2} and trivially \eqref{Ncond1}. 
Since $A^+$ and $-A^-$ are the identity on their domains, $R=P^- -LP^+$ and $A^{\mathrm{inv}}=\alpha^3$. 
The IBC \eqref{IBC} becomes
\be
R \psi^{(1)}(x_1,x_2,0) = -\tfrac{i}{\hbar} R \alpha^3 N(x_1,x_2) \, \psi^{(0)}\,.
\ee
Multiplying both sides from the left by $-\gamma^0-i$ (which yields an equivalent equation because $-\gamma^0-i$ is invertible because $i$ is not an eigenvalue of $-\gamma^0$ because $\gamma^0$ is self-adjoint)
yields the IBC \eqref{ex1ibc} of Section~\ref{sec:ex1}, as $(-\gamma^0-i)R=\gamma^3-i$.$\hfill\square$

\bigskip

\noindent{\bf Example~9.} Lienert and Nickel \cite{LN18} studied a model of particle creation in one space dimension. We now explain how it fits into the general scheme we have presented. 

Their massless particles move in $\RRR^1$, and they split and coalesce according to $x \leftrightarrows x+x$. For their model, they prove existence and uniqueness of solutions of the evolution (even multi-time) for suitably smooth initial wave functions. Their Hamiltonian would be local on full Fock space, but to avoid technical difficulties they truncated the Fock space at a maximal particle number. For simplicity, we limit ourselves to the case of two sectors that they discussed in their section~3. The configuration space is $\Q=\Q^{(1)} \cup \Q^{(2)}$ with $\Q^{(1)}=\RRR$, $\Q^{(2)}=\{(z_1,z_2)\in \RRR^2: z_1\leq z_2\}$, and $\partial \Q =\partial \Q^{(2)}= \{(z,z):z\in\RRR\}$. The Riemann metric on $\Q$ is just the Euclidean metric of $\RRR^1$ and $\RRR^2$; $\mu^{(n)}$ is just $n$-dimensional volume, and $\lambda(\{(z,z):z\in B\})= \sqrt{2} \mu^{(1)}(B)$. The jump mapping is $f(z,z)=z$; the set $f^{-1}(q)$ always has either 0 or 1 element (0 for $q\in \Q^{(2)}$, 1 for $q\in \Q^{(1)}$); for $q=z\in \Q^{(1)}$, $\nu_z$ is $\sqrt{2}$ times the counting measure, $\nu_z\{(z,z)\}=\sqrt{2}$, and $n(z,z)=(-1/\sqrt{2},1/\sqrt{2})$ is the inward-pointing unit normal vector. We write the components of $\psi^{(n)}$ as
\be
\psi^{(1)} = \begin{pmatrix} \psi_- \\ \psi_+ \end{pmatrix} \in \CCC^2,~~~~~\psi^{(2)}= \begin{pmatrix} \psi_{--}\\\psi_{-+}\\\psi_{+-}\\\psi_{++} \end{pmatrix} \in \CCC^2 \otimes \CCC^2\,.
\ee
As in Example~7, we use the 1d Dirac equation; $A^n$ is given by \eqref{L15An} and $LP^+$ and $R$ by \eqref{L15LP+} at all $(z,z)\in\partial \Q$. 
However, the boundary condition \eqref{BCL} of Example~7 now gets replaced with an IBC of the form
\be\label{IBCLN}
\psi_{-+}^{(2)}(z,z) - e^{i\theta}\psi^{(2)}_{+-}(z,z) = B\, \psi^{(1)}(z)
\ee
with $B$ a certain $1\times 2$ matrix. Their Hamiltonian is of the form \eqref{Hdef}, explicitly 
\begin{align}
(H\psi)^{(1)}(z) &= -i\alpha^1 \partial_z \psi^{(1)}(z) + \sqrt{2}\, N(z)^\dagger \, \psi^{(2)}(z,z)\\
(H\psi)^{(2)}(z_1,z_2) &= (-i\alpha_1^1 \partial_1 -i\alpha_2^1 \partial_2) \psi^{(2)}(z_1,z_2)
\end{align}
with $N(z)=N$ (independent of $z$) a certain linear mapping $\CCC^2\to \CCC^2\otimes \CCC^2$ (a $4\times 2$ matrix) that Lienert and Nickel called $-\tfrac{1}{\sqrt{2}} A^\dagger$, see Eq.s~(23) and (24) in \cite{LN18}. Condition \eqref{Ncond1} requires that the first and last rows of $N$ vanish, in agreement with Eq.~(35) of \cite{LN18}. Condition \eqref{Ncond2} amounts to $N^\dagger \diag(0,-1,1,0) N=0$ or $-N_{2i}^* N_{2k}+N_{3i}^* N_{3k}=0$ for $i,k=1,2$, which is equivalent to Eq.~(36) of \cite{LN18}. 
We further obtain that $A^{\mathrm{inv}}=2^{-1/2}\diag(0,-1,+1,0)$, and the IBC \eqref{IBC} prescribes that the matrix $B$ on the right-hand side of \eqref{IBCLN} is given by 
\be
B = -\tfrac{i}{\hbar}2^{-1/4}R A^{\mathrm{inv}} N = \tfrac{i}{\sqrt{2}\hbar} (N_{21}+e^{i\theta}N_{31}, N_{22}+e^{i\theta}N_{32})\,,
\ee
which, considering $N=-\tfrac{1}{\sqrt{2}} A^\dagger$, agrees with Eq.~(37) of \cite{LN18}.\hfill$\square$

\subsection{General IBC}
\label{sec:full}

The expression \eqref{Hdef} for the action of the Hamiltonian is already the general one. The IBC \eqref{IBC}, however, will now be replaced by a wider class of conditions. Since an IBC is a linear relation between $\psi(q)$ with $q\in\partial\Q$ and $\psi(f(q))$, it corresponds to a subspace $\tS_q$ of 
\be\label{tEdef}
\tE_q := E_q \oplus E_{f(q)}
\ee
in the sense that the boundary condition reads
\be\label{bctS}
\bigl(\psi(q),\psi(f(q)) \bigr) \in \tS_q \quad \forall q\in\partial\Q\,.
\ee
The space $\tS_q$ must be so chosen as to ensure conservation of probability and self-adjointness of the Hamiltonian. So let us investigate the conservation of probability.

\label{sec:conserv}

As the balance equation of $|\psi|^2$, we obtain from \eqref{Hdef} that
\begin{align}
\partial_t |\psi(q)|^2
&= \tfrac{2}{\hbar}\, \Im \Bigl(\psi(q) \Big| H\psi(q) \Bigr)_{\!q} \\
&= -\sum_{a=1}^{d_n} \nabla_{\!a} j^a(q) 
+  \!\!\! \int\limits_{f^{-1}(q)} \!\!\! \nu_q(dq') \, \tfrac{2}{\hbar}\, \Im \Bigl(\psi(q) \Big| N(q')^\dagger \psi(q') \Bigr)_{\!q} \,.
\label{gain}
\end{align}
The first term (divergence of the current) represents the transport of probability by continuous flow within one sector $\Q^{(n)}$, whereas the second term times $\mu(dq)$ equals the amount of probability gained per time in the volume element $dq$ around $q$ due to the last term in the Hamiltonian \eqref{Hdef}. The amount of probability lost per time due to flux into the part $f^{-1}(dq)$ of the boundary $\partial\Q$ is equal to
\be\label{loss}
\text{loss} = -\!\!\! \int\limits_{f^{-1}(dq)}\!\!\! \lambda(dq') \, j^n(q')\,.
\ee
Conservation of probability means that gain = loss, or that $\mu(dq)$ times the last term in \eqref{gain} equals \eqref{loss}. Expressing $\lambda$ using \eqref{munulambda}, dividing by $\mu(dq)$, and letting $dq\to\{q\}$, we obtain that conservation of probability is equivalent to
\be\label{gainloss}
\int\limits_{f^{-1}(q)} \!\!\! \nu_q(dq') \, \tfrac{2}{\hbar}\, \Im \Bigl(\psi(q) \Big| N(q')^\dagger \psi(q') \Bigr)_{\!q}
= -\!\!\! \int\limits_{f^{-1}(q)}\!\!\! \nu_q(dq') \, j^n(q')
\ee
for all $q$. Now dropping the integration relative to $\nu_q(dq')$, interchanging the names $q\leftrightarrow q'$, and using $j^n=(\psi|A^n\psi)$, we obtain that local conservation of probability means that, for every $q\in\partial\Q$,
\be\label{sufficient1}
 \tfrac{2}{\hbar}\, \Im \Bigl(\psi(f(q)) \Big| N(q)^\dagger \psi(q) \Bigr)_{\!f(q)} = - \Bigl( \psi(q) \Big| A^n(q)\, \psi(q) \Bigr)_{\!q}\,.
\ee

We introduce a notation adapted to this situation: Let $\psi_{\pm}=P^{\pm}\psi(q)$, $\psi_0=P^0\psi(q)$, $\psi_*=\psi(f(q))$, $N_{\pm}=P^{\pm}N(q)$, $N_0=P^0N(q)$,  
\be
\tpsi= \begin{pmatrix}\psi_+\\\psi_- \\ \psi_0\\ \psi_*\end{pmatrix} \in \tE_{q}= E_{q}\oplus E_{f(q)}
\ee
with $\oplus$ meaning orthogonal sum, and
\be\label{tAdef}\def\arraystretch{1.8}
\tA = \left[ \begin{array}{c|c|c|c} 
A^+ &0&0& \tfrac{i}{\hbar}N_+\\[1mm]\hline
0&A^-&0&\tfrac{i}{\hbar}N_-\\[1mm]\hline
0&0&0&\tfrac{i}{\hbar}N_0\\[1mm]\hline
-\tfrac{i}{\hbar} N_+^\dagger & -\tfrac{i}{\hbar} N_-^\dagger & -\tfrac{i}{\hbar} N_0^\dagger &0
\end{array} \right] \,.
\ee
Then the condition \eqref{sufficient1} for local conservation of probability can be abbreviated as
\be\label{tsufficient1}
\tpsi^\dagger \tA \tpsi =0\,.
\ee
Let $\tKegel$ denote the set of $\tpsi\in\tE$ satisfying \eqref{tsufficient1}. To ensure conservation of probability, we need that $\tS_{q} \subseteq \tKegel$, but for self-adjointness of the Hamiltonian, we need a bit more, for essentially the same reasons as discussed around Conjecture~\ref{conj:refl} in Section~\ref{sec:Lagrangian}.

\begin{conj}\label{conj:IBC}
Let $\tS_q\subseteq \tE_q$ be a subbundle. The interior--boundary condition
\be\label{IBCconj}
\tpsi(q) := \bigl( \psi(q), \psi(f(q)) \bigr) \in \tS_q \quad \forall q\in\partial \Q
\ee
can occur in a self-adjoint extension (from $C_c^\infty(E|_{\Q^\circ})$) of $H$ as in \eqref{Hdef} in $\Lz(E)$ with local conservation of probability as in \eqref{sufficient1} iff $\tS_q$ is a complete Lagrangian subspace of $\tE_q$ relative to $\tA$ for every $q\in \partial \Q$.
\end{conj}

Put differently, we suggest that the complete Lagrangian property of $\tS_q$ is the algebraic (or formal) condition relevant to the self-adjointness of $H$ with IBC \eqref{IBCconj} with local conservation of probability.

\begin{prop}\label{prop:ibc}
Let $\hE:= E^0_q \oplus E_{f(q)}$, $\hA:\hE \to \hE$ the endomorphism given by
\be\label{hAdef}\def\arraystretch{1.5}
\hA = \left[ \begin{array}{c|c}
0& \tfrac{i}{\hbar} N_0 \\[1mm]\hline
-\tfrac{i}{\hbar} N_0^\dagger & -\tfrac{1}{\hbar^2} N^\dagger A^{\mathrm{inv}}N
\end{array} \right]\,,
\ee
$\hP^0:=1_{\{0\}}(\hA)$, $\hP^+:=1_{(0,\infty)}(\hA)$, and $\hP^-:=1_{(-\infty,0)}(\hA)$; let $\hE^0$, $\hE^+$, $\hE^-$ be their respective ranges, and let $\hA^{\pm}$ be the restriction of $\hA$ to $\hE^{\pm}$. The complete Lagrangian subspaces $\tS_q$ of $\tE_q$ relative to $\tA$ are in a natural bijective relation to the unitary isomorphisms
\be\label{tL}
\tL: E^+\oplus\hE^+\to E^-\oplus\hE^-\,, 
\ee
given by
\begin{multline}\label{tStL}
\tS_q = \Biggl\{ \tpsi\in\tE_q: 
\begin{pmatrix} (-A^-)^{1/2} \psi_- -\tfrac{i}{\hbar} (-A^-)^{-1/2} N_- \psi_*\\ (-\hA^-)^{1/2} \hP^- \tpsi \end{pmatrix} = \\
\tL \begin{pmatrix} (A^+)^{1/2} \psi_+ + \tfrac{i}{\hbar} (A^+)^{-1/2} N_+ \psi_* \\ (\hA^+)^{1/2} \hP^+ \tpsi \end{pmatrix} \Biggr\}\,.
\end{multline}
In particular, complete Lagrangian subspaces exist iff
\be\label{dimh+-}
\dim E^+ + \dim\hE^+= \dim E^- +\dim \hE^-\,.
\ee
\end{prop}

We give the proof in Appendix~\ref{app:proofs}. Note that, in particular, every $\tS_q$ of the form \eqref{tStL} is a subset of $\tKegel$, which implies that probability is conserved, i.e.,
\be\label{probcons}
\partial_t |\psi(q)|^2
= -\sum_{a=1}^{d_n} \nabla_{\!a} j^a(q) 
 -\!\!\! \int\limits_{f^{-1}(q)}\!\!\! \nu_q(dq') \, j^n(q') \,.
\ee

Here is how the IBC \eqref{IBC} of Section~\ref{sec:ibc} is included in Proposition~\ref{prop:ibc} as a special case: If, as demanded in \eqref{Ncond1} and \eqref{Ncond2}, $N_0=0$ and $N^\dagger A^{\mathrm{inv}} N=0$, then $\hA=0$, so $\hP^+=0=\hP^-$, and $\tL$ reduces to a unitary isomorphism $L:E^+\to E^-$ as in the case of reflecting boundary conditions. The equation inside the set brackets defining $\tS$ in \eqref{tStL} becomes
\be
(-A^-)^{1/2} \psi_- -\tfrac{i}{\hbar} (-A^-)^{-1/2} N_- \psi_*=
L\Bigl( (A^+)^{1/2} \psi_+ + \tfrac{i}{\hbar} (A^+)^{-1/2} N_+ \psi_*\Bigr)\,,
\ee
which is equivalent to \eqref{IBC}.

A final remark. Could there be further boundary conditions that conserve probability? One could imagine that \eqref{gainloss} is satisfied but \eqref{sufficient1} is not; that is, that the integrals in \eqref{gainloss} are equal but the integrands are not. When visualized using Bohmian trajectories, such a situation presumably corresponds to jumps from boundary points to other boundary points, instead of from one sector to another. This case can occur in the situation of Theorem~\ref{thm:ex2} and is relevant, e.g., to $\delta$ potentials concentrated on a surface $S$ in $\RRR^3$ \cite{DES89,B16}, where the two sides of $S$ are regarded as two separated boundaries but a transition from one side to the same point on the other side is possible, while the values of $\psi$ on the two sides are related by a transmission condition. However, this case will not be considered in this section.

\subsection{Upshot}
\label{sec:summary}

Let us collect the central equations. 
The Hamiltonian acts according to
\be\label{Hdefr}
H\psi(q) = -i\hbar\sum_{a=1}^{d_n} A^a(q) \nabla_{\!a}\psi(q) + B(q)\psi(q)
+ \!\!\! \int\limits_{f^{-1}(q)} \!\!\! \nu_q(dq') \, N(q')^\dagger\, \psi(q')\,,
\ee
where $\nu_q$ can be thought of as the uniform measure over $f^{-1}(q)$, $N(q):E_{f(q)}\to E_q$ is a given field of linear mappings, $A^a(q)$ is self-adjoint, and the skew-adjoint part of $B(q)$ is $-\tfrac{\hbar}{2}\nabla_{\!a}A^a(q)$.

The general form of the IBC that we have derived (and that we conjecture to be the most general possible IBC with local conservation of probability for Dirac-type equations) is
\be\label{IBCg}
\begin{pmatrix} (-A^-)^{1/2} P^-\psi_q -\tfrac{i}{\hbar} (-A^-)^{-1/2} P^-N \psi_{f(q)}\\ (-\hA^-)^{1/2} \hP^- (P^0\psi_q,\psi_{f(q)}) \end{pmatrix} = 
\tL \begin{pmatrix} (A^+)^{1/2} P^+ \psi_q + \tfrac{i}{\hbar} (A^+)^{-1/2} P^+N \psi_{f(q)} \\ (\hA^+)^{1/2} \hP^+ (P^0\psi_q,\psi_{f(q)}) \end{pmatrix}
\ee
at every boundary point $q$, with a fixed unitary isomorphism $\tL$. Here, $P^+=1_{(0,\infty)}(A^n)$, $A^n= n(q) \cdot A(q)$, $P^-=1_{(-\infty,0)}(A^n)$, $A^\pm = A|_{\text{range }P^\pm}$, $A^{\mathrm{inv}}= (A^+)^{-1}P^+ + (A^-)^{-1} P^-$ and likewise with the ``hatted'' mappings obtained from $\hA$, the endomorphism of $\hE:= E^0_q \oplus E_{f(q)}$ given by 
\be\label{hAdef2}\def\arraystretch{1.5}
\hA = \left[ \begin{array}{c|c}
0& \tfrac{i}{\hbar} P^0 N \\[1mm]\hline
-\tfrac{i}{\hbar} (P^0N)^\dagger & -\tfrac{1}{\hbar^2} N^\dagger A^{\mathrm{inv}}N
\end{array} \right].
\ee

In case $\hA=0$, the IBC \eqref{IBCg} reduces to
\be\label{IBCr}
R(q) \, \psi(q) = -\tfrac{i}{\hbar} \, R(q) \, A^{\mathrm{inv}}(q)\, N(q) \, \psi\bigl(f(q)\bigr)
\ee
at every boundary point $q$. Here, $R$ is defined as
\be
R=\sqrt{-A^-}P^- - L\sqrt{A^+}P^+\,,
\ee
where $L:\text{range }P^+ \to \text{range }P^-$ is a fixed unitary isomorphism.

\appendix

\section{Proof of Self-Adjointness}
\label{app:sa}

In this appendix, we prove Theorem~\ref{thm:ex2}. Let $\Omega \subset \RRR^3$ be a domain with compact $C^2$-boundary $\Sigma:=\partial \Omega$ and let $\vn(\vx)$ be the inward pointing unit normal vector field. Our Hilbert space is $\Hilbert = \Hilbert_0 \oplus \Hilbert_1 := \CCC \oplus \Lz(\Omega, \CCC^4)$. From now on, we will write $\Lz(M)$ instead of $\Lz(M, \CCC^4)$ for any measurable $M \subset \RRR^n$. We denote for any $s \in \RRR$ by $H^s(M) = H^s(M, \CCC^4)$ the corresponding vectorial Sobolev spaces. Let $D:= - i \valpha \cdot \nabla + m \beta$ denote the action of the Dirac operator (for $\hbar =1$) and let $N \in H^{1/2}(\Sigma)$ be a spinor field with the property \eqref{Nalphaex2}, i.e.,
\begin{align}
\label{eq:propofN}
\langle N , \alpha^{\vn} N \rangle_{\Lz(\Sigma)}= 0 \, .
\end{align}
The matrix $\B(\vx):=- i \beta \alpha^{\vn}(\vx)= - i \gamma^{\vn}(\vx)$ is Hermitian with eigenvalues $\pm 1$. We define $P_+ := \frac{1+\B}{2}$ and $P_- := \frac{1-\B}{2}$, the (orthogonal) projections onto the eigenspaces. In the following, $\psi_0$ and $\psi_1$ will be the components of $\psi$ in $\Hilbert_0$ and $\Hilbert_1$, and the \emph{trace} of $\psi_1$ will mean the restriction of $\psi_1$ to $\Sigma$ (the common terminology in connection with the Sobolev embedding theorem) and be denoted by $\tr\,\psi_1$. The IBC \eqref{ex2ibc} can equivalently be rewritten as
\be\label{ex2ibc2}
P_-   \tr\, \psi_1=- i P_- \alpha^{\vn}  N \psi_0\,.
\ee

Consider the operator $H$ with action 
\be
H (\psi_0, \psi_1) := \bigl(\langle N, \tr \psi_1\rangle_{\Lz(\Sigma)}, D \psi_1 \bigr)
\ee
on the domain 
\begin{align}
\label{eq:defofD}
\domain(H)=\Bigl\{ \psi \in \CCC \oplus \Lz(\Omega) : \psi_1 \in H^1(\Omega), \, \eqref{ex2ibc2}~\text{holds a.e.\ on }\Sigma\Bigr\} \, .
\end{align}  
The operator $H$ is well defined on $\domain(H)$ because the Dirac operator $D$ maps $H^1(\Omega)$ to $\Lz(\Omega)$ and the trace operator $\tr$ maps $H^1(\Omega)$ to $\Lz(\Sigma)$ by the Sobolev embedding theorem.

\begin{prop}
$\domain(H)$ is dense in $\Hilbert$.
\end{prop}

\begin{proof}
The trace operator $\tr$ extends to a continuous operator from $H^1(\Omega)$ \textit{onto} $H^{1/2}(\Sigma)$, that is, there exists an extension operator $E:H^{1/2}(\Sigma) \to H^1(\Omega)$ such that $\tr \circ E = 1_{H^{1/2}(\Sigma)}$, see \cite[Prop.~1.1]{OV18}. For any $\psi_0 \in \CCC$ we set $\mathcal{D}_{\psi_0} := \{ \psi_1  = E(- i \alpha^{\vn} N \psi_0)+ f \vert f \in C_0^\infty(\Omega) \}$. Elements in $\mathcal{D}_{\psi_0}$ lie in the first Sobolev space and fulfill the interior-boundary condition \eqref{ex2ibc2}. The set $\mathcal{D}_{\psi_0}$ is dense in $\Hilbert_1$ for any $\psi_0 \in \Hilbert_0$. Thus, $\domain(H)$ is dense in $\Hilbert$.
\end{proof}

Let $\mathfrak{D}(\Omega):= \{\psi \in \Lz(\Omega) \vert D \psi \in \Lz(\Omega) \}$ be the domain of the maximal Dirac operator. It is known \cite[2.1]{OV18} that $\tr$ can be extended to an operator from $\mathfrak{D}(\Omega)$ into $H^{-1/2}(\Sigma)$, and that for $\phi \in \mathfrak{D}(\Omega)$ and $\psi \in H^1(\Omega)$, the following (dual) Green's identity holds \cite[2.15]{OV18}:
\begin{align}
\label{eq:greens}
\langle \phi_1, D \psi_1 \rangle_{\Hilbert_1} - \langle D \phi_1,  \psi_1 \rangle_{\Hilbert_1}  = i \langle \tr \phi_1 , \alpha^{\vn} \tr \psi_1 \rangle_{H^{-1/2}(\Sigma),H^{1/2}(\Sigma)} \, .
\end{align}
Here, $\langle \cdot, \cdot \rangle_{H^{-1/2}(\Sigma),H^{1/2}(\Sigma)} =: \langle \cdot, \cdot \rangle_{-\frac 1 2, \frac 1 2}$ denotes the pairing between elements of $H^{1/2}(\Sigma)$ and its dual space $H^{-1/2}(\Sigma)$. 
Because $\beta \alpha^{\vn} = - \alpha^{\vn} \beta$, we have that $P_+ \alpha^{\vn} = \alpha^{\vn} P_-$. Using this, we can rewrite \eqref{eq:greens} as
\begin{align}
\langle& \phi_1, D \psi_1 \rangle_{\Hilbert_1} - \langle D \phi_1,  \psi_1 \rangle_{\Hilbert_1}  = i  \langle \tr \phi_1 , \alpha^{\vn} \tr \psi_1 \rangle_{{-\frac 1  2},{\frac 1  2}}
 \nonumber \\
 &
  =  i \langle P_- \tr \phi_1 ,  \alpha^{\vn} P_+ \tr \psi_1 \rangle_{H^{-1/2},H^{1/2}}  +  i \langle P_+ \tr \phi_1 ,  \alpha^{\vn} P_- \tr \psi_1 \rangle_{{-\frac 1  2},{\frac 1  2}}
  \nonumber  \\ \label{eq:boundaryform}
 &
  = \langle P_+ \tr \phi_1 ,   i \alpha^{\vn} P_- \tr \psi_1 \rangle_{{-\frac 1  2},{\frac 1  2}}  - \langle i \alpha^{\vn} P_- \tr \phi_1 ,  P_+ \tr \psi_1 \rangle_{{-\frac 1  2},{\frac 1  2}}    \, .
\end{align}

\begin{prop}\label{prop:sym}
The operator $(H,\domain(H))$ is symmetric.
\end{prop}

\begin{proof}
This is a straightforward computation starting from~\eqref{eq:boundaryform} using the interior-boundary condition \eqref{ex2ibc2} and the assumption~\eqref{eq:propofN}.
\end{proof}

Now we collect some of the tools that we plan to use in the proof of self-adjointness of $H$ later on.
\begin{itemize}
\item Let $\psi_1 \in \mathfrak{D}(\Omega)$. If $\tr \psi_1 \in H^{1/2}(\Sigma)$, then $\psi_1 \in H^1(\Omega)$. \cite[Prop.~2.16]{OV18}
\item There are two operators on $H^{-1/2}(\Sigma)$, called $C_\pm$, such that $f= C_+(f) + C_-(f)$ for all $f \in H^{-1/2}(\Sigma)$. \cite[Prop.~2.6]{OV18}
\item The operators $C_\pm$ map $H^{1/2}(\Sigma)$ onto itself. \cite[Prop.~2.6]{OV18}
\item The composition $C_- \circ \tr$ maps $\mathfrak{D}(\Omega)$ into $H^{1/2}(\Sigma)$. \cite[Prop.~2.7]{OV18}
\item It holds that $C_+(\mathfrak{B} f) = - \mathfrak{B} \left(C_-(f)+ i \mathcal{A}(f)\right)$ for all $f \in H^{-1/2}(\Sigma)$. Here $\mathcal{A}$ is an operator on $H^{-1/2}(\Sigma)$ that maps into $H^{1/2}(\Sigma)$. \cite[Prop.~2.6,~2.8]{OV18}
\end{itemize}

With the help of these facts, the proof of self-adjointness of $H$ can be carried out using the strategies outlined in \cite{OV18}.

\begin{prop}
\label{prop:sa}
The operator $H$ is self-adjoint.
\end{prop}
\begin{proof}
Due to Proposition~\ref{prop:sym}, it remains to prove that $(H, \domain(H))$ extends $(H^*,\domain(H^*))$.

Let $\phi \in \domain(H^*)$. By definition, there exists $H^* \phi:=\eta \in \Hilbert$ such that $\langle \phi, H \psi \rangle_\Hilbert = \langle \eta, \psi \rangle_\Hilbert$ for all $\psi \in \domain(H)$.
In a first step, we take $\psi_0=0$ and $\psi_1 = f \in C_0^\infty(\Omega)$. For this choice of $\psi$ it holds that
\be
\langle \eta_1 , f \rangle = \langle \phi_1, D f \rangle \qquad \forall f \in C_0^\infty(\Omega) \, ,
\ee
which is nothing but the definition of the distributional Dirac operator, so $\eta_1 = D \phi_1$ as distributions. Since $\eta_1 \in \Hilbert_1$, we can conclude that $\eta_1 = D \phi_1$ as functions and therefore $\phi_1 \in \mathfrak{D}(\Omega)$. This implies that 
\begin{align}
&\langle \eta_0, \psi_0\rangle_{\Hilbert_0} +  \langle \eta_1, \psi_1\rangle_{\Hilbert_1} = \langle \phi_0, \langle N, \tr\, \psi_1\rangle_{\Lz(\Sigma)} \rangle_{\Hilbert_0} + \langle \phi_1, D \psi_1 \rangle_{\Hilbert_1}\nonumber
\\
\iff& \langle \eta_0, \psi_0\rangle_{\Hilbert_0} = \langle \phi_0, \langle N, \tr\, \psi_1\rangle \rangle_{\Hilbert_0} + \langle \phi_1, D \psi_1 \rangle_{\Hilbert_1} -  \langle D \phi_1, \psi_1\rangle_{\Hilbert_1} \, .
\end{align}
By~\eqref{eq:boundaryform}, we can rewrite this as
\begin{align}
& \langle \phi_0, \langle N, \tr \psi_1\rangle \rangle_{\Hilbert_0}
 + \langle P_+ \tr \phi_1 ,   i \alpha^{\vn} P_- \tr \psi_1 \rangle_{{-\frac 1  2},{\frac 1  2}}  - \langle i \alpha^{\vn} P_- \tr \phi_1 ,  P_+ \tr \psi_1 \rangle_{{-\frac 1  2},{\frac 1  2}} \nonumber \\ \label{eq:eta0}
 &
 =\langle \eta_0, \psi_0\rangle  \, .
\end{align}
In the next step, we keep $\psi_0=0$ and use the extension operator to obtain $\psi_1 = E(P_+ f) \in H^1(\Omega)$ for $f \in H^{1/2}(\Sigma)$. As a result of~\eqref{eq:eta0}, for all $f \in H^{1/2}(\Sigma)$, we have that
\begin{align}
0&= \langle \phi_0, \langle N, P_+ f\rangle \rangle_{\Hilbert_0}   -  \langle i \alpha^{\vn} P_- \tr \phi_1 ,   P_+ f \rangle_{{-\frac 1  2},{\frac 1  2}}
\\
 \implies 0&=\langle P_+ N \phi_0 - P_+ i \alpha^{\vn}  \tr \phi_1, f \rangle_{{-\frac 1  2},{\frac 1  2}} \, .
\end{align}
As a consequence, $P_+ N \phi_0 = P_+ i \alpha^{\vn}  \tr \phi_1$ in $H^{-1/2}(\Sigma)$. But since we have assumed that $N \in H^{1/2}(\Sigma)$, and since $P_+$ is smooth enough, see \cite[Rem.~2.5]{OV18}, this equality actually holds in $H^{1/2}(\Sigma)$. We can rewrite it as
\begin{align}
P_+ N \phi_0 &= P_+ i \alpha^{\vn}  \tr \phi_1 \nonumber\\
\iff \alpha^{\vn} P_- \alpha^{\vn}  N \phi_0 &=  \alpha^{\vn} i P_-   \tr \phi_1 \nonumber\\
\iff~~~~~~~~  P_-   \tr \phi_1 &=   - i P_- \alpha^{\vn}  N \phi_0  \,,
\end{align}
which means that $\phi$ satisfies the IBC \eqref{ex2ibc2}. Now we know that $P_- \tr \phi_1 \in H^{1/2}(\Sigma)$. In order to use \cite[Prop.~2.16]{OV18} and conclude that $\phi \in \domain(H)$, we still have to show that $P_+ \tr \phi_1 \in H^{1/2}(\Sigma)$. First, observe that
\begin{align}
\label{eq:cminus}
C_-(P_+ \tr \phi_1)= C_-( \tr \phi_1)- C_-(P_- \tr \phi_1) \in H^{1/2}(\Sigma)
\end{align}
because $C_- \circ \tr$ maps $\mathfrak{D}(\Omega)$ into $H^{1/2}(\Sigma)$ and we already know that $P_- \tr \phi_1 \in H^{1/2}(\Sigma)$ while $C_-$ maps it back into the same space. 
Second, we have that
\begin{align}
C_+(P_+ \tr \phi_1) 
&= C_+((1+\B) \tr \phi_1) = C_+(\B P_+ \tr \phi_1) 
\nonumber\\
&= -\mathfrak{B} \left(C_-(P_+ \tr \phi_1)+ i \mathcal{A}(P_+ \tr \phi_1)\right) 
\end{align}
Due to~\eqref{eq:cminus}, the first term is regular. For compact $\Sigma$, the operator $\mathcal{A}$ maps $H^{-1/2}(\Sigma)$ into $H^{1/2}(\Sigma)$. The equality $f= C_+(f) + C_-(f)$ then yields that $P_+ \tr \phi_1 \in H^{1/2}(\Sigma)$.

Thus, we can conclude that $\tr \phi_1 \in H^{1/2}(\Sigma)$, and that all dual pairings in \eqref{eq:eta0} are in fact inner products. Using the interior-boundary condition \eqref{ex2ibc2} both for $\psi$ and $\phi$ together with~\eqref{eq:propofN}, we arrive at
\begin{align*}
& \langle \phi_0, \langle N, \tr \psi_1\rangle \rangle_{\Hilbert_0}
 + \langle P_+ \tr \phi_1 ,   i \alpha^{\vn} P_- \tr \psi_1 \rangle_{{\Lz(\Sigma)}}  - \langle i \alpha^{\vn} P_- \tr \phi_1 ,  P_+ \tr \psi_1 \rangle_{\Lz(\Sigma)}
\\
&
=  \langle i N  \phi_0,  \alpha^{\vn} P_+ N \psi_0 \rangle_{\Lz(\Sigma)}  +  \langle P_+ \tr \phi ,      N \psi_0 \rangle_{{\Lz(\Sigma)}}
\\
&
=  \langle i N  \phi_0,  \alpha^{\vn}  N \psi_0 \rangle_{\Lz(\Sigma)} + \langle - i P_- \alpha^{\vn} N  \phi_0,   N \psi_0 \rangle_{\Lz(\Sigma)}  +  \langle P_+ \tr \phi_1 ,      N \psi_0 \rangle_{{\Lz(\Sigma)}}
\\
&
=   \langle P_- \tr  \phi_1,   N \psi_0 \rangle  +  \langle P_+ \tr \phi_1 ,      N \psi_0 \rangle_{\Lz(\Sigma)} =  \langle \langle N, \tr \phi_1\rangle ,      \psi_0 \rangle_{\Hilbert_0} =\langle \eta_0, \psi_0\rangle_{\Hilbert_0} \, .
\end{align*}
This finally yields $\eta_0 =\langle N, P_+ \tr \phi_1\rangle$. We have thus proved that $(H, \domain(H))$ extends $(H^*,\domain(H^*))$.
\end{proof}

This completes the proof of Theorem~\ref{thm:ex2}.

\section{Proofs of Propositions 1 and 2}
\label{app:proofs}

\begin{proof}[Proof of Proposition~\ref{prop:refl}]
Suppose that $S$ is a complete Lagrangian subspace. Then $E^0\subseteq S^\#$, so $E^0\subseteq S$ and thus $S=E^0 \oplus S'$, where $\oplus$ means orthogonal sum. We claim that $P^+|_{S'}$ is bijective $S'\to E^+$. 

Here is why it must be injective: If $\phi_1,\phi_2\in S'$ and $P^+\phi_1=P^+ \phi_2$, then $\chi:=\phi_1-\phi_2\in S'$ and $P^+\chi=0$. It follows that $0=(\chi|A\chi)=(P^-\chi|A^-P^-\chi)= -\|(-A^-)^{1/2} P^-\chi \|^2$, so $P^-\chi=0$ and thus $\chi \in E^0 \perp S'$, so $\chi=0$. 

Here is why it must be surjective: Otherwise, there would be $\phi\in E^+$ with $\phi\neq 0$ and $\phi \perp P^+ \chi$ for every $\chi \in S'$. As a consequence, $\phi \perp P^+ \chi$ even for every $\chi\in S$. Set $\phi' = (A^+)^{-1} \phi$. Then, for every $\chi\in S$, $(\phi'|A\chi)= ((A^+)^{-1}\phi|A^+ P^+\chi)= (\phi|P^+ \chi)=0$. Hence, $\phi'\in S$; in particular, $(\phi'|A\phi')=0$. But also $\phi'\in E^+$, so $0=(\phi'|A^+\phi')= \|(A^+)^{1/2}\phi' \|^2$ and therefore $\phi'=0$, in contradiction to the assumption.

Likewise, $P^-|_{S'}$ must be bijective $S'\to E^-$. Set
\be
L := (-A^-)^{-1/2} P^-|_{S'}  (P^+|_{S'})^{-1} (A^+)^{-1/2} \,.
\ee
Then $L:E^+ \to E^-$ is bijective and unitary because, for every $\phi\in E^+$ and $\chi = (P^+|_{S'})^{-1} (A^+)^{-1/2}\phi$ (which lies in $S'$),
\begin{subequations}
\begin{align}
0&=(\chi|A\chi)\\ 
&= \bigl(P^+\chi \big| A^+ P^+ \chi \bigr) + \bigl(P^-\chi \big| A^-P^-\chi \bigr)\\
&= \bigl( (A^+)^{-1/2}\phi \big| A^+ (A^+)^{-1/2} \phi \bigr) - (L\phi| L\phi)\\
&= (\phi|\phi) - (L\phi| L\phi)\,.
\end{align}
\end{subequations}
This shows that $S$ must be of the form \eqref{SL}. 

Conversely, \eqref{SL} defines a subspace $S$. To see that it is complete Langrangian, note that for $\phi\in E$ and $\chi\in S$, 
\begin{align}
(\phi|A\chi) 
&= \bigl( P^+\phi \big| A^+ P^+ \chi \bigr) - \bigl( P^-\phi \big| (-A^-) P^- \chi \bigr)\\
&= \bigl( P^+\phi \big| A^+ P^+ \chi \bigr) - \Bigl( P^-\phi \Big|  (-A^-)^{1/2}  L(A^+)^{1/2}P^+ \chi \Bigr)\\
&= \Bigl( A^+ P^+\phi - (A^+)^{1/2}L^\dagger(-A^-)^{1/2}P^-\phi \Big|  P^+ \chi \Bigr)\,.
\end{align}
Under which condition on $\phi$ does this expression vanish for all $\chi\in S$? Since $P^+\chi$ can be chosen arbitrarily in $E^+$, this happens only when
\be
A^+ P^+\phi - (A^+)^{1/2}L^\dagger(-A^-)^{1/2}P^-\phi=0
\ee
or, equivalently (multiply by $L(A^+)^{-1/2}$), $L(A^+)^{1/2} P^+\phi =(-A^-)^{1/2} P^- \phi$, which means $\phi\in S$. 
\end{proof}

\bigskip

\begin{proof}[Proof of Proposition~\ref{prop:ibc}]
Define the endomorphism $\tF:\tE\to\tE$ by
\be\label{tFdef}\def\arraystretch{1.8}
\tF = \left[ \begin{array}{c|c|c|c} 
~~~I~~~&0&0& -\tfrac{i}{\hbar} (A^+)^{-1}N_+\\[1mm]\hline
0&~~~I~~~&0&-\tfrac{i}{\hbar} (A^-)^{-1} N_-\\[1mm]\hline
0&0&~~~I~~~&0\\[1mm]\hline
0&0&0&I
\end{array} \right]
\ee
and set $A' = \tF^\dagger \tA \tF$. Then $\tS$ is complete Lagrangian relative to $\tA$ iff $S':= \tF^{-1} \tS$ is complete Lagrangian relative to $A'$. One finds that
\be\label{tF-1}\def\arraystretch{1.8}
A'=\left[ \begin{array}{c|c|c|c} 
A^+&0&0&0\\[1mm]\hline
0&~A^-&0&0\\[1mm]\hline
0&0&0&\tfrac{i}{\hbar} N_0\\[1mm]\hline
0&0&-\tfrac{i}{\hbar} N_0^\dagger& -\tfrac{1}{\hbar^2} N^\dagger A^{\mathrm{inv}} N
\end{array} \right], \quad
\tF^{-1}= \left[ \begin{array}{c|c|c|c} 
I~&0&0& \tfrac{i}{\hbar} (A^+)^{-1}N_+\\[1mm]\hline
0&~I~&0&\tfrac{i}{\hbar} (A^-)^{-1} N_-\\[1mm]\hline
0&0&~I~&0\\[1mm]\hline
0&0&0&I
\end{array} \right] .
\ee
Note that $A'$ is self-adjoint and that the lower right 4 blocks of $A'$ coincide with $\hA$, so with respect to the decomposition $\tE=E^+_q\oplus E^-_q \oplus \hE^+ \oplus \hE^- \oplus \hE^0$ we can write $A'$ in the form
\be\label{A'}\def\arraystretch{1.4}
A'=\left[ \begin{array}{c|c|c|c|c} 
A^+&0&0&0&0\\\hline
0&~A^-&0&0&0\\\hline
0&0&~\hA^+&0&0\\\hline
0&0&0&~\hA^-&0\\\hline
0&0&0&0&~0~
\end{array} \right].
\ee
We can read off that the positive spectral subspace of $A'$ is $E^+ \oplus \hE^+$, its negative spectral subspace is $E^- \oplus \hE^-$, and its kernel is $\hE^0$. Using Proposition~\ref{prop:refl}, we can read off further that the complete Lagrangian subspaces relative to $A'$ are the subspaces of the form
\be\label{S'}
S'= \Biggl\{\tpsi\in\tE: \begin{pmatrix} (-A^-)^{1/2} \psi_-\\ (-\hA^-)^{1/2} \hP^- \tpsi \end{pmatrix}  =\tL \begin{pmatrix} (A^+)^{1/2} \psi_+ \\ (\hA^+)^{1/2}\hP^+ \tpsi \end{pmatrix} \Biggr\}
\ee
with unitary isomorphism $\tL: E^+ \oplus \hE^+ \to E^- \oplus \hE^-$,
hence the complete Langrangian subspaces relative to $\tA$ are the subspaces of the form \eqref{tStL}.
\end{proof}

\bigskip

\noindent\textit{Acknowledgments.} 
We are grateful to Jonas Lampart, Matthias Lienert, and Lukas Nickel for helpful discussions. This work was supported by the German Research Foundation (DFG) within the Research Training Group 1838 \textit{Spectral Theory and Dynamics of Quantum Systems}.

\end{document}